\documentclass[11pt, a4paper, reqno]{article}

\usepackage{a4wide}
\usepackage[english]{babel}
\usepackage[utf8]{inputenc}   
\usepackage{amsmath,amsfonts,amssymb,amsthm}
\usepackage{lmodern}
\usepackage{microtype}
\usepackage{enumerate}
\usepackage{tikz}
\usepackage{hyperref}
\hypersetup{colorlinks=true,citecolor=blue,linkcolor=blue,urlcolor=blue}

\usetikzlibrary{calc}

\DeclareMathOperator{\reducesto}{\leq_{SA}}
\DeclareMathOperator{\proj}{\ensuremath{\pi}}
\DeclareMathOperator{\fpol}{fPol}
\DeclareMathOperator{\supp}{{\rm supp}}
\DeclareMathOperator{\CSP}{CSP}
\DeclareMathOperator{\VCSP}{VCSP}

\DeclareMathOperator{\feas}{Feas}
\DeclareMathOperator{\opt}{Opt}
\DeclareMathOperator{\rdom}{{Feas}}
\DeclareMathOperator{\pol}{Pol}
\DeclareMathOperator{\Var}{Var}
\DeclareMathOperator*{\E}{\mathop{{\mathbb E}}}
\DeclareMathOperator{\avg}{{\mathrm avg}}
\DeclareMathOperator{\ar}{ar}

\newcommand{\N}{\mbox{$\mathbb N$}}
\newcommand{\eq}[1]{\ensuremath{\phi^{#1}_{=}}}
\newcommand{\qq}{\ensuremath{\overline{\mathbb{Q}}}}
\newcommand{\inst}{\ensuremath{I}}
\newcommand{\tup}[1]{\ensuremath{\mathbf #1}}
\newcommand{\multiset}[1]{\ensuremath{\{\hspace*{-2pt}\{#1\}\hspace*{-2pt}\}}}
\newcommand{\vcspopt}[0]{{\rm Opt}}
\newcommand{\vcspval}[0]{{\rm Val}}
\renewcommand{\vec}[1]{\ensuremath{\mathbf{#1}}}

\newcommand{\lpopt}[3]{{\rm Opt_{LP}}(#1)}
\newcommand{\lpval}[4]{{\rm Val_{LP}}(#1,#2)}

\newcommand\Crestrict[2]{{
  \left.\kern-\nulldelimiterspace 
  #1 
  \right|_{#2} 
  }}
 
\newtheorem{theorem}{Theorem}[section]
\newtheorem{lemma}[theorem]{Lemma}
\newtheorem{proposition}[theorem]{Proposition}
\newtheorem{corollary}[theorem]{Corollary}
\newtheorem*{lemma*}{Lemma} 
\newtheorem*{proposition*}{Proposition} 
\newtheorem*{theorem*}{Theorem} 
\theoremstyle{definition}
\newtheorem{definition}[theorem]{Definition}
\newtheorem{example}[theorem]{Example}
\newtheorem{remark}[theorem]{Remark}

\begin{document}
\title{The power of Sherali-Adams relaxations\\ for general-valued CSPs\footnote{An extended
abstract of part of this work appeared in the \emph{Proceedings of the 42nd International Colloquium on Automata, Languages, and
Programming (ICALP'15)}~\cite{tz15:icalp}.}}

\author{
Johan Thapper\\
Universit\'e Paris-Est, Marne-la-Vall\'ee, France\\
\texttt{thapper@u-pem.fr}
\and
Stanislav \v{Z}ivn\'{y}\thanks{Stanislav \v{Z}ivn\'y was supported by a Royal
Society University Research Fellowship. Part of this work was done when the
second author was visiting the Simons Institute for the Theory of Computing at UC Berkeley.
This project has received funding from the European Research Council (ERC) under
the European Union's Horizon 2020 research and innovation programme (grant
agreement No 714532). The paper reflects only the authors' views and not the
views of the ERC or the European Commission. The European Union is not liable
for any use that may be made of the information contained therein.}\\
University of Oxford, UK\\
\texttt{standa.zivny@cs.ox.ac.uk}
}

\date{}
\maketitle

\begin{abstract} 

We give a precise algebraic characterisation of the power of Sherali-Adams
relaxations for solvability of valued constraint satisfaction problems to
optimality. The condition is that of bounded width which has already been shown
to capture the power of local consistency methods for decision CSPs and the
power of semidefinite programming for robust approximation of CSPs. 

Our characterisation has several algorithmic and complexity consequences. On the
algorithmic side, we show that several novel and many known valued constraint
languages are tractable via the third level of the Sherali-Adams relaxation. For
the known languages, this is a significantly simpler algorithm than the
previously obtained ones. On the complexity side, we obtain a dichotomy theorem
for valued constraint languages that can express an injective unary function.
This implies a simple proof of the dichotomy theorem for conservative
valued constraint languages established by Kolmogorov and \v{Z}ivn\'y [JACM'13],
and also a dichotomy theorem for the exact solvability of Minimum-Solution
problems. These are generalisations of Minimum-Ones problems to \emph{arbitrary}
finite domains. Our result improves on several previous classifications by
Khanna et al. [SICOMP'00], Jonsson et al. [SICOMP'08], and Uppman [ICALP'13].

\end{abstract}

%
%
\section{Introduction}

Convex relaxations are one of the most powerful techniques for designing
polynomial-time exact and approximation
algorithms~\cite{Chlamtac12:sdp,BarakS14}. The idea is to formulate the problem
at hand as an integer program and relax it to a convex program which can be
solved in polynomial time, such as a linear program (LP) or a semidefinite
program (SDP). A solution to the problem is then obtained by designing a
(possibly randomised) polynomial-time algorithm that converts the solution to
such a relaxation into an integer solution to the original problem.

Convex relaxations can be strengthened by including additional constraints which
are satisfied by an integer solution. This process of generating stronger
relaxations by adding larger (but still local) constraints is captured by
various hierarchies of convex relaxations, including the hierarchy of linear
programming relaxations proposed by Sherali and Adams~\cite{Sherali1990}, that
by Lov\'asz and Schrijver~\cite{Lovasz91:jo}, and their semidefinite programming
versions, including the hierarchy of Lasserre~\cite{Lasserre01:jo} (see
also~\cite{Laurent03:mor} for a nice comparison of these hierarchies). For an
integer program with $n$ variables taking values in $\{0,1\}$, the convex
program obtained by $n$ levels of any of the above-mentioned hierarchies has
integrality gap $1$, that is, it gives an exact solution (but the program may
take exponential time to solve). Since the size of a program obtained by $k$
levels of these hierarchies is $n^{O(k)}$, for a constant $k$, the program can
be solved in polynomial time.

In this paper we study \emph{constant level} Sherali-Adams relaxations for
\emph{exact} solvability of discrete optimisation problems. We do this within
the framework of constraint satisfaction problems, which captures a large family
of both theoretical and practical problems. An instance of the \emph{valued
constraint satisfaction problem} (VCSP) is given by a collection of variables
that is assigned labels from a given finite domain with the goal to
\emph{minimise} an objective function given by a sum of weighted relations (cost
functions), each depending on some subset of the
variables~\cite{Cohen06:complexitysoft}. The weighted relations can take on
finite rational values and positive infinity. 

By varying the codomain of the weighted relations, we get a variety of
interesting problems. When the codomain is $\{0,\infty\}$, we get the class of
decision problems known as \emph{constraint satisfaction
problems}~\cite{Feder98:monotone} with the goal to determine whether or not
there is a labelling for all variables that evaluates the objective function to
zero. When the codomain is $\{0,1\}$, we get the class of optimisation problems
known as \emph{minimum} constraint satisfaction
problems~\cite{Creignou95:maximum,Creignouetal:siam01,Khanna00:approximability}.
When the codomain is $\mathbb{Q}$, we get the class of optimisation problems
known as \emph{finite-valued} (or generalised~\cite{Raghavendra08:stoc})
constraint satisfaction problems~\cite{tz16:jacm}. The special case of having a
domain of size two has been studied extensively under the name of pseudo-Boolean
optimisation~\cite{Boros:pseudo-boolean,Crama11:book}. Finally, by allowing a
codomain to be both $\mathbb{Q}$ and positive infinity, we get the large class
of problems known as \emph{valued} constraint satisfaction
problems~\cite{Cohen06:complexitysoft,kz17:survey}. Intuitively, the infinite
value deems certain labellings forbidden and thus all constraints are required
to be satisfied, whereas the rational values model the optimisation aspect of
the problem. 

We remark that this framework is more general than that of \emph{mixed} CSPs
with hard and soft constraints used in the approximation
community~\cite{Kumar11:soda}, where each constraint is either hard or soft;
hard constraints correspond to $\{0,\infty\}$-valued weighted relations in our
framework, and a soft constraint corresponds to a $\{0,w\}$-valued weighted
relation, where $w$ is the weight of the constraint.  Thus, all constraints in
mixed CSPs are $2$-valued.

Valued CSPs are sometimes also called \emph{general-valued} CSPs to emphasise
the fact that (decision) CSPs are a special case of valued CSPs.

For constraint satisfaction problems, an important algorithmic technique is
\emph{local consistency methods}, i.e.\ considering a bounded number of
variables at a time and propagating infeasible partial assignments. Problems for
which such techniques suffice to decide satisfiability are said to have
\emph{bounded width}. In an important series of papers,
\cite{Maroti08:weakly,LaroseZadori07:au,Bulatov09:width,Barto14:jacm}, the
property of having bounded width has been shown to be equivalent to a
universal-algebraic condition, now known as the ``bounded width condition''.
There is a clear relation between the local propagation in consistency methods
for decision CSPs and the consistent marginals-condition of Sherali-Adams
relaxations. In this paper, we demonstrate the applicability of powerful
universal-algebraic techniques, developed for decisions CSPs, in the study of
linear programming hierarchies for valued CSPs.

\subsection*{Contributions}

A set $\Gamma$ of weighted relations on some fixed finite domain is called a
\emph{valued constraint language}.
We denote by $\VCSP(\Gamma)$ the class of VCSP 
instances with all weighted relations from $\Gamma$. 

In our first result, we give an
algebraic~\cite{Bulatov05:classifying,cccjz13:sicomp}
characterisation of the power of Sherali-Adams relaxations for VCSPs.
Theorem~\ref{thm:main}, presented in Section~\ref{sec:sapower}, shows that for a valued constraint language $\Gamma$ of
finite size the following three statements are equivalent: 
(i) $\Gamma$ is tractable via a constant level Sherali-Adams relaxation;
(ii) $\Gamma$ is tractable via the third level Sherali-Adams relaxation;
(iii) the support clone of $\Gamma$ contains (not necessarily idempotent) 
$m$-ary weak near-unanimity operations for every $m\geq 3$.\footnote{The precise definition of weak near-unanimity operations can be found in
Section~\ref{sec:prelims}.}
The condition (iii) is precisely that of ``bounded width'' for constraint
languages with codomain $\{0,\infty\}$ (such languages are known as crisp)~\cite{Maroti08:weakly,LaroseZadori07:au,Bulatov09:width,Barto14:jacm}.
Note that the implication ``(ii)\ $\Longrightarrow$\ (i)'' is trivial.

The implication ``(iii)\ $\Longrightarrow$\ (ii)'', proved in
Section~\ref{sec:suff}, is shown via linear programming duality and
fundamentally relies on~\cite{Barto14:jacm} and~\cite{Barto16:jloc}. 
This result
simplifies and generalises several previously obtained tractability results for
valued constraint languages, as discussed in Section~\ref{sec:algs}. For
example, valued constraint languages with a tournament pair multimorphism were
previously known to be tractable using ingenious application of various
consistency techniques, advanced analysis of constraint networks using modular
decompositions, and submodular function
minimisation~\cite{Cohen08:Generalising}. Here, we show that an even less
restrictive condition (having a binary conservative commutative operation in
some fractional polymorphism) ensures that the third level of the Sherali-Adams
relaxation solves all instances to optimality.

The implication ``(i)\ $\Longrightarrow$\ (iii)'', proved in
Section~\ref{sec:nec}, is shown by proving that, given a language $\Gamma$ that
violates (iii), $\Gamma$ can simulate linear equations in some Abelian group.
This result is known for $\{0,\infty\}$-valued constraint
languages~\cite{Barto14:jacm}. It suffices to show that linear equations can
fool constant level Sherali-Adams relaxations, which is proved in
Section~\ref{sec:gap}, and that the ``simulation'' preserves bounded level of
Sherali-Adams relaxations for valued constraint languages, which is proved in
Section~\ref{sec:reductions}. Previously, it was only known that this
``simulation'' preserves polynomial-time reducibility.
One immediate corollary of our result is a classification of conservative
valued constraint languages~\cite{kz13:jacm} \emph{without} relying
on~\cite{Takhanov10:stacs}. In fact we give an alternative and still simple
proof of the complexity classification of conservative valued constraint
languages~\cite{kz13:jacm}, which implies that tractable conservative valued
constraint languages are captured by a majority operation in the
support clone, which was not previously known.

Overall, we give a precise characterisation of the power of Sherali-Adams
relaxations for exact solvability of VCSPs. This rather surprising result
demonstrates how robust the concept of bounded width is, capturing not only the
power of local consistency methods for decision
CSPs~\cite{Bulatov09:width,Barto14:jacm,Bulatov16:lics} and the class of decision CSPs that can
be robustly approximated~\cite{Barto16:sicomp}, but also the power of
Sherali-Adams relaxations for exact solvability of VCSPs.

Minimum-Solution~\cite{Jonsson08:max-sol} problems are special types of VCSPs
that involve $\{0,\infty\}$-valued weighted relations together with a single
unary $\mathbb{Q}$-valued weighted relation that is required to be injective.
(The natural encoding of Vertex Cover as a VCSP instance is of this kind.)
Minimum-Solution problems include integer programming over bounded domains and
can be viewed as a generalisation of Min-Ones
problems~\cite{Creignouetal:siam01,Khanna00:approximability} to larger domains.
Compare this to the result~\cite{ccjkpz16:arxiv-binarisation} that any VCSP
instance is equivalent to a VCSP instance with only binary relations and
unary (not necessarily injective) finite-valued weighted relations. Hence,
unless we settle the CSP dichotomy conjecture~\cite{Feder98:monotone}, some additional requirement on
the unary weighted relations (such as injectivity) is necessary.

As a corollary of our characterisation, we give, in Section~\ref{sec:compl}, a
complete complexity classification of exact solvability of Minimum-Solution
problems over \emph{arbitrary finite} domains, thus improving on previous
partial classifications for domains of size two~\cite{Khanna00:approximability}
and three~\cite{Uppman13:icalp}, homogeneous and maximal (under a certain
algebraic conjecture) languages~\cite{Jonsson08:siam} and on graphs with few
vertices~\cite{Jonsson07:maxsol}. Theorem~\ref{thm:min-sol} shows that the
Minimum-Solution problem is NP-hard unless it satisfies the bounded width
condition. Previous partial results included ad-hoc algorithms for
various special cases. Our result shows that one algorithm, the third level of
the Sherali-Adams relaxation, solves all tractable cases and is thus universal.
As a matter of fact, we actually prove a complexity classification for a larger
class of problems that includes Minimum-Solutions problems as a special case, as
described in detail in Section~\ref{sec:compl}.

\subsection*{Related work}

The first level of the Sherali-Adams hierarchy is known as the \emph{basic
linear programming} (BLP) relaxation~\cite{Chekuri04:sidma}.
In~\cite{tz12:focs}, the authors gave a precise algebraic
characterisation of $\Gamma$ for which any instance of $\VCSP(\Gamma)$ is solved
to optimality by BLP, see also~\cite{ktz15:sicomp}.
The characterisation proved important not only in the study of
VCSPs~\cite{hkp14:sicomp} and other classes of problems~\cite{Hirai16:mp}, but
also in the design of fixed-parameter algorithms~\cite{Iwata16:sicomp}.
In~\cite{tz16:jacm}, it was then shown that for finite-valued CSPs, the BLP
solves \emph{all} tractable cases; i.e.\ if BLP fails to solve any instance of
some finite-valued constraint language then this language is NP-hard.
The BLP has been considered in the context of CSPs for robust
approximability~\cite{Kun12:itcs,Dalmau13:robust} and constant-factor
approximation~\cite{Ene13:soda,Dalmau15:soda}. 
Higher levels of Sherali-Adams hierarchy have been considered for
(in)approximability of CSPs~\cite{Vega07:soda,Chan:jacm,Yoshida14:itcs}.
Semidefinite programming relaxations have also been considered in the context of
CSPs for approximability~\cite{Raghavendra08:stoc} and robust
approximability~\cite{Barto16:sicomp}.
Concrete lower bounds on Sherali-Adams and other relaxations
include~\cite{Schoenebeck07:ccc,Charikar09:talg,Georgiou10:sicomp,Alekhnovich11:toc}.
Whilst the complexity of valued constraint languages is open, it has been shown
that a dichotomy for constraint languages, conjectured
in~\cite{Feder98:monotone}, implies a dichotomy for valued constraint
languages~\cite{Kolmogorov15:focs}. Our results give a complete complexity
classification for a large class of VCSPs without any dependence on the
dichotomy conjecture~\cite{Feder98:monotone}. Since the announcement of our
results~\cite{tz15:icalp}, the tractability results obtained in this paper were
shown using different methods (preprocessing combined with an LP
relaxation)~\cite{Kolmogorov15:focs}.

One ingredient of our proof is the fact that constant level Sherali-Adams
relaxations cannot solve exactly instances involving equations over a
non-trivial Abelian group. This is known to follow, via~\cite{Tulsiani09:stoc},
from a stronger result of Grigoriev~\cite{Grigoriev01:tcs}, later rediscovered
by Schoenebeck~\cite{Schoenebeck08:focs}, that limits the power of $\Omega(n)$ levels of
Lasserre SDP relaxations for approximately solving Max-CSPs involving equations. However, a formal
proof would require the definition of SDP relaxations that are not in the scope of
this article. Rather, we provide here a direct, elementary proof of this fact
and observe that our proof actually gives a gap instance for Sherali-Adams
relaxations of level $\Theta(\sqrt{n})$. This also has the advantage of our
proof being self-contained.

%
%
\section{Preliminaries}
\label{sec:prelims}

\subsection{Valued CSPs}

We denote by $[m]$ the set $\{1,2,\ldots,m\}$. Let $\qq=\mathbb{Q}\cup\{\infty\}$ denote the set of rational numbers extended
with positive infinity. Throughout the paper, let $D$ be a fixed finite set of
size at least two, also called a \emph{domain}; we call the elements of $D$
\emph{labels}.

\begin{definition}\label{def:wrel}
An $r$-ary \emph{weighted relation} over $D$ is a mapping $\phi:D^r\to\qq$. We
write $\ar(\phi)=r$ for the arity of $\phi$.
\end{definition}

A weighted relation $\phi \colon D^r\to\qq$ is called \emph{finite-valued} if
$\phi(\tup{x}) < \infty$ for all $\tup{x} \in D^r$.
A weighted relation $\phi \colon D^r \to \{0,\infty\}$ can be seen as the (ordinary) relation
$\{ \tup{x} \in D^r \mid \phi(\tup{x}) = 0 \}$.
We will use both viewpoints interchangeably.

For any $r$-ary weighted relation $\phi$, we denote by
$\feas(\phi)=\{\tup{x}\in D^r \mid \phi(\tup{x})<\infty\}$ the
underlying $r$-ary \emph{feasibility relation}, and by
$\opt(\phi)=\{\tup{x}\in\feas(\phi) \mid \forall\tup{y}\in D^r:
\phi(\tup{x})\leq\phi(\tup{y})\}$ the $r$-ary \emph{optimality relation},
which contains the tuples on which $\phi$ is minimised.

\begin{definition}
Let $V=\{x_1,\ldots, x_n\}$ be a set of variables. A \emph{valued constraint} over $V$ is an expression
of the form $\phi(\tup{x})$ where $\phi$ is a weighted relation and $\tup{x}\in V^{\ar(\phi)}$.
The tuple $\tup{x}$ is called the \emph{scope} of the constraint.
\end{definition}

We will use the notational convention to denote by $X_i$
the \emph{set} of variables occurring in the scope $\tup{x}_i$.

\begin{definition}\label{def:vcsp}
An instance $I$ of the \emph{valued constraint satisfaction problem} (VCSP) is specified
by a finite set $V=\{x_1,\ldots,x_n\}$ of variables, a finite set $D$ of labels,
and an \emph{objective function} $\phi_\inst$
expressed as follows:
\begin{equation}
\phi_\inst(x_1,\ldots, x_n)=\sum_{i=1}^q{\phi_i(\tup{x}_i)}\,,
\label{eq:sepfun}
\end{equation}
where each $\phi_i(\tup{x}_i)$, $1\le i\le q$, is a valued constraint. 
Each constraint may appear multiple times in~$\inst$.
An \emph{assignment} to $\inst$ is a map $\sigma \colon V \to D$.
The goal is to find an assignment that minimises the objective function.
\end{definition}

For a VCSP instance $I$, we write $\vcspval(\inst,\sigma)$ for
$\phi_\inst(\sigma(x_1), \dots, \sigma(x_n))$, and $\vcspopt(\inst)$ for the
minimum of $\vcspval(\inst,\sigma)$ over all assignments.

An assignment $\sigma$ with $\vcspval(\inst,\sigma)<\infty$ is called
\emph{satisfying}. A VCSP instance $I$ is called \emph{satisfiable} if there is
a satisfying assignment to $I$. CSPs are a special case of VCSPs with
(unweighted) relations with the goal to determine the existence of a satisfying
assignment.

A \emph{valued constraint language}, or just a constraint language, over $D$ is
a set of weighted relations over $D$. We denote by $\VCSP(\Gamma)$ the class of
all VCSP instances in which the weighted relations are all contained in
$\Gamma$. A constraint language $\Delta$ is called \emph{crisp} if $\Delta$
contains only (unweighted) relations. For a crisp language $\Delta$ we denote by
$\CSP(\Delta)$ the class $\VCSP(\Delta)$ to emphasise the fact that there is no
optimisation involved.

A valued constraint language $\Gamma$ is called \emph{tractable} if
$\VCSP(\Gamma')$ can be solved (to optimality) in polynomial time for every
finite subset $\Gamma'\subseteq\Gamma$, and $\Gamma$ is called
\emph{NP-hard} if $\VCSP(\Gamma')$ is NP-hard for some finite
$\Gamma'\subseteq\Gamma$.

\subsection{Fractional Polymorphisms}

Given an $r$-tuple $\tup{x}\in D^r$, we denote its $i$th entry by $\tup{x}[i]$ for $1\leq i\leq r$.
A mapping $f \colon D^m\rightarrow D$ is called an $m$-ary \emph{operation} on $D$; $f$ is
\emph{idempotent} if $f(x,\ldots,x)=x$.
We apply an $m$-ary operation $f$ to $m$ $r$-tuples
$\tup{x_1},\ldots,\tup{x_m}\in D^r$ coordinatewise, that is, \begin{equation}
f(\tup{x_1},\ldots,\tup{x_m})=(f(\tup{x_1}[1],\ldots,\tup{x_m}[1]),\ldots,f(\tup{x_1}[r],\ldots,\tup{x_m}[r]))\,.
\end{equation}

\begin{definition} \label{def:pol}
Let $\phi$ be a weighted relation on $D$ and let $f$ be an $m$-ary operation on $D$.
We call $f$ a \emph{polymorphism of $\phi$} if,
for any $\tup{x_1},\ldots,\tup{x_m} \in \rdom(\phi)$,
we have that $f(\tup{x_1},\ldots,\tup{x_m})\in\rdom(\phi)$.

For a valued constraint language $\Gamma$,
we denote by $\pol(\Gamma)$ the set of all operations which are polymorphisms of all 
$\phi \in \Gamma$. We write $\pol(\phi)$ for $\pol(\{\phi\})$.
\end{definition}

A probability distribution $\omega$ over the set of $m$-ary operations on $D$ is called
an $m$-ary \emph{fractional operation}.
We define $\supp(\omega)$ to be the set of operations assigned positive probability by $\omega$.

The following two notions are known to capture the complexity of valued
constraint languages~\cite{cccjz13:sicomp,Kozik15:icalp} and will also be important in this
paper.

\begin{definition} \label{def:wp} 
Let $\phi$ be a weighted relation on $D$ and
let $\omega$ be an $m$-ary fractional operation on $D$.
We call $\omega$ a \emph{fractional polymorphism of $\phi$} if 
$\supp(\omega)\subseteq\pol(\phi)$ and for any
$\tup{x}_1,\ldots,\tup{x}_m \in \rdom(\phi)$, we have
\begin{equation}
\E_{f\sim \omega}[\phi(f(\vec{x}_1,\ldots,\vec{x}_m))]\ \le\
\avg\{\phi(\vec{x}_1),\ldots,\phi(\vec{x}_m)\}\,.
\label{eq:wpol}
\end{equation}
For a valued constraint language $\Gamma$, we denote by $\fpol(\Gamma)$ the set of all
fractional operations which are fractional polymorphisms of all weighted
relations $\phi \in \Gamma$. 
We say that $\Gamma$ is \emph{improved} by $\omega$ if $\omega\in\fpol(\Gamma)$.
We write $\fpol(\phi)$ for $\fpol(\{\phi\})$.
\end{definition}

\begin{example}
Consider the domain $D=\{0,1\}$ and the two binary operations $\min$ and $\max$
on $D$ that return the smaller and the larger its two arguments, respectively. A
valued constraint language on $D$ is called \emph{submodular} if it has the
fractional polymorphism $\omega$ defined by
$\omega(\min)=\omega(\max)=\frac{1}{2}$.
\end{example}

\begin{definition}
Let $\Gamma$ be a valued constraint language on $D$. We define
\begin{equation}
\supp(\Gamma)\ =\ \bigcup_{\omega \in \fpol(\Gamma)} \supp(\omega)\,.
\end{equation}
\end{definition}

An $m$-ary \emph{projection} is an operation of the form
$\proj^{(m)}_i(x_1,\ldots,x_m)=x_i$ for some $1\leq i\leq m$. Projections are
polymorphisms of all valued constraint languages. 

The \emph{composition} of an $m$-ary operation $f:D^m\rightarrow D$ with $m$
$n$-ary operations $g_i:D^n\rightarrow D$ for $1\leq i\leq m$ is the
$n$-ary function $f[g_1,\ldots,g_m]:D^n\to D$ defined by
\begin{equation}
f[g_1,\ldots,g_m](x_1,\ldots,x_n)=f(g_1(x_1,\ldots,x_n),\ldots,g_m(x_1,\ldots,x_n))\,.
\end{equation}

A \emph{clone} of operations is a set of operations on $D$ that contains
all projections and is closed under composition. 
$\pol(\Gamma)$ is a clone for any valued constraint language $\Gamma$.

\begin{lemma}\label{lem:suppclone}
For any valued constraint language $\Gamma$, $\supp(\Gamma)$ is a clone.
\end{lemma}

We note that Lemma~\ref{lem:suppclone} has also been observed
in~\cite{fz16:toct} and in~\cite{Kozik15:icalp}. For completeness, we give a
proof here. (Our proof is slightly different from the proofs
in~\cite{fz16:toct,Kozik15:icalp} as we have defined fractional
polymorphisms as probability distributions.)

\begin{proof}
Observe that $\supp(\Gamma)$ contains all projections as
$\tau_m\in\fpol(\Gamma)$ for every $m\geq 1$, where $\tau_m$ is the fractional
operation defined by $\tau_m(\proj^{(m)}_i)=\frac{1}{m}$ for every $1\leq i\leq
m$. Thus we only need to show that $\supp(\Gamma)$ is closed under composition.

Let $f\in\supp(\Gamma)$ be an $m$-ary operation with $\omega(f)>0$ for some
$\omega\in\fpol(\Gamma)$. Moreover, let $g_i\in\supp(\Gamma)$ be
$n$-ary operations with $\mu_i(g_i)>0$ for some $\mu_i \in \fpol(\Gamma)$,
where $1\leq i\leq m$.
We define an $n$-ary fractional operation 
\begin{align}
\omega'(p)\ =\ \Pr_{\substack{t \sim \omega\\ h_i \sim \mu_i}} \left[ t[h_1, \ldots, h_m] = p \right]\,.
\end{align}

Since $\omega(f)>0$ and $\mu_i(g_i)>0$ for all $1\leq i\leq m$, we have
$\omega'(f[g_1,\ldots,g_m])>0$. 
A straightforward verification shows that $\omega'\in\fpol(\Gamma)$.
Consequently, $f[g_1,\ldots,g_m]\in\supp(\Gamma)$.
\end{proof}

The following lemma is a generalisation of~\cite[Lemma~2.9]{tz16:jacm} from
arity one to arbitrary arity and from finite-valued to valued constraint
languages, but the proof is analogous. A special case has also been observed, in
the context of Min-Sol problems~\cite{Uppman13:icalp}, by Hannes
Uppman~\cite{Uppman15:thesis}.

\begin{lemma}\label{lem:killing}
Let $\Gamma$ be a valued constraint language of finite size on a domain $D$ and let $f \in  \pol(\Gamma)$.
Then, $f \in \supp(\Gamma)$ if, and only if, $f \in \pol(\opt(\phi_I))$ for all
instances $I$ of $\VCSP(\Gamma)$.
\end{lemma}

\begin{proof}
Let $m$ be the arity of $f$.
The operation $f$ is in $\supp(\Gamma)$ if, and only if, there exists a fractional polymorphism
$\omega$ with $f \in \supp(\omega)$.
This is the case if, and only if, the following system of linear inequalities in 
the variables $\omega(g)$ for $m$-ary $g \in \pol(\Gamma)$ is satisfiable:
\begin{align}
  \sum_{g\in\pol(\Gamma)} \omega(g) \phi(g(\tup{x}_1,\dots,\tup{x}_m)) & \leq 
  \avg\{\phi(\tup{x}_1),\dots,\phi(\tup{x}_m)\}, \quad \forall \phi \in \Gamma, \tup{x}_i \in \feas(\phi), \notag \\
  \sum_{g\in\pol(\Gamma)} \omega(g) & = 1, \notag \\ 
  \omega(f) & > 0, \notag \\
  \omega(g) & \geq 0, \quad \forall \mbox{$m$-ary\ }g\in\pol(\Gamma).
  \label{eq:killing1}
\end{align}

By Farkas' lemma (e.g.~\cite[Lemma~2.8]{tz16:jacm}), the system (\ref{eq:killing1}) is unsatisfiable if, and only if, the following
system in variables $z(\phi,\tup{x}_1,\dots,\tup{x}_m)$, for $\phi \in \Gamma, \tup{x}_i \in \feas(\phi)$,
is satisfiable:
\begin{align}
  \sum_{\substack{\phi \in \Gamma\\ \tup{x}_i \in \feas(\phi)}}
  z(\phi,\tup{x}_1, \dots, \tup{x}_m) \left( \avg\{\phi(\tup{x}_1),\dots,\phi(\tup{x}_m)\} 
  -\phi(g(\tup{x}_1,\dots,\tup{x}_m)) \right)  
  & \leq 0,
  \quad \forall \mbox{$m$-ary\ }g \in \pol(\Gamma), \notag \\
  \sum_{\substack{\phi \in \Gamma\\ \tup{x}_i \in \feas(\phi)}}
  z(\phi,\tup{x}_1,\dots,\tup{x}_m) \left( \avg\{\phi(\tup{x}_1),\dots,\phi(\tup{x}_m)\} 
  - \phi(f(\tup{x}_1, \dots, \tup{x}_m)) \right)
  & < 0, \notag \\
  z(\phi, \tup{x}_1, \dots, \tup{x}_m) & \geq 0, \quad \forall \phi \in \Gamma, \tup{x}_i \in \feas(\phi).
  \label{eq:killing2}
\end{align}

First, assume that $f \not\in \supp(\Gamma)$ so that (\ref{eq:killing2}) has a feasible solution $z$.
Note that by scaling we may assume that $z$ is integral.
Let $V^{(m)} = \{ v_{\tup{x}} \mid \tup{x} \in D^m \}$ and
let $\tup{v} = (v_1, \dots, v_n)$ be an enumeration of $V^{(m)}$.
Define $\iota \colon V^{(m)} \to D^m$ by $\iota(v_{\tup{x}}) = \tup{x}$
and
let $I$ be the instance of $\VCSP(\Gamma)$ with variables $V^{(m)}$
and objective function:
\[
\phi_{I}(\tup{v}) = \sum_{\substack{\phi \in \Gamma\\ \tup{x}_i \in \feas(\phi)}}
z(\phi, \tup{x}_1, \dots, \tup{x}_m)
\phi(\iota^{-1}(\tup{x}_1[1],\dots,\tup{x}_m[1]), \dots, \iota^{-1}(\tup{x}_1[\ar(\phi)],\dots,\tup{x}_m[\ar(\phi)])),
\]
where the multiplication by $z$ is simulated by taking the corresponding constraint with multiplicity $z$.
According to (\ref{eq:killing2}), 
every projection $\proj^{(m)}_i$ induces
an optimal assignment $\proj^{(m)}_i \circ \iota$ to $I$.
Interpreted as $D^m$-tuples, we therefore have $\proj^{(m)}_i \in \opt(\phi_{I})$ for $1 \leq i \leq m$.
On the other hand, (\ref{eq:killing2}) states that $f \circ \iota$ is not an optimal assignment,
so $f(\proj^{(m)}_1, \dots, \proj^{(m)}_m) \not\in \opt(\phi_{I})$.
In other words, $f \not\in \pol(\opt(\phi_{I}))$, and $I$ is an instance
of $\VCSP(\Gamma)$.

For the opposite direction, assume that $f \in \supp(\Gamma)$,
so that (\ref{eq:killing2}) is unsatisfiable.
Let $I$ be an instance of $\VCSP(\Gamma)$ with objective function 
$\phi_I(y_1, \dots, y_n) = \sum_p \phi_p(\tup{y}_p)$.
Let $\sigma_1, \dots, \sigma_m \in \opt(\phi_I)$.
We will consider $\sigma_j$ both as tuples and as assignments $V \to D$.
In particular, $\sigma_j(\tup{y}_p)$ is the projection of $\sigma_j$ onto the scope $\tup{y}_p$.
Let $z(\phi, \tup{x}_1, \dots, \tup{x}_m)$ be the number of indices $p$ for which $\phi = \phi_p$
and $\sigma_j(\tup{y}_p) = \tup{x}_j$ for every $1 \leq j \leq m$.
Then,
\begin{align*}
\sum_{\substack{\phi \in \Gamma\\ \tup{x}_i \in \feas(\phi)}}
  z(\phi,\tup{x}_1, \dots, \tup{x}_m)  \avg\{\phi(\tup{x}_1),\dots,\phi(\tup{x}_m)\} 
&=
\sum_{p} \avg_j \{\phi_p(\sigma_j(\tup{y}_p)) \} \\
&= 
\avg_j \{ \sum_p \phi_p(\sigma_j(\tup{y}_p)) \}
=
\vcspopt(I)
\end{align*}
and, for all $g \in \pol(\Gamma)$,
\[
\sum_{\substack{\phi \in \Gamma\\ \tup{x}_i \in \feas(\phi)}}
  z(\phi,\tup{x}_1, \dots, \tup{x}_m) \phi(g(\tup{x}_1, \dots, \tup{x}_m))
  =
  \sum_p \phi_p(g(\sigma_1(\tup{y}_p), \dots, \sigma_m(\tup{y}_p))).
  \]
It follows that all non-strict inequalities in (\ref{eq:killing2}) are satisfied by $z$,
and since (\ref{eq:killing2}) is unsatisfiable, this implies that
$\vcspopt(I) \leq  \sum_p \phi_p(f(\sigma_1(\tup{y}_p), \dots, \sigma_m(\tup{y}_p))$
must hold with equality
so $f(\sigma_1, \dots, \sigma_m) \in \opt(\phi_I)$.
Since the $\sigma_j$ were chosen arbitrarily, $f \in \pol(\opt(\phi_I))$.
This establishes the lemma.
\end{proof}

\subsection{Cores and Constants}
\label{sec:cores}

\begin{definition}
Let $\Gamma$ be a valued constraint language with domain $D$ and let $S \subseteq D$.
The \emph{sub-language $\Gamma[S]$ of $\Gamma$ induced by $S$} is the valued constraint
language defined on domain $S$ and containing the restriction of every weighted relation
$\phi\in\Gamma$ onto $S$.
\end{definition}

\begin{definition}\label{def:core}
A valued constraint language $\Gamma$ is \emph{a core} if all unary operations in
$\supp(\Gamma)$ are bijections.
A valued constraint language $\Gamma'$ is a \emph{core of $\Gamma$} if $\Gamma'$ is a
core and $\Gamma' = \Gamma[f(D)]$ for some unary $f \in \supp(\Gamma)$.
\end{definition}

The following lemma implies that when studying the computational complexity of a
valued constraint language $\Gamma$‚ we may assume that $\Gamma$ is a core.

\begin{lemma}\label{lem:core}
Let $\Gamma$ be a valued constraint language and $\Gamma'$ a core of $\Gamma$.
Then, for all instances $I$ of $\VCSP(\Gamma)$ and $I'$ of $\VCSP(\Gamma')$, where
$I'$ is obtained from $I$ by substituting each weighted relation in $\Gamma$ for
its restriction in $\Gamma'$, the optimum of $I$ and $I'$ coincide.
\end{lemma}

A special case of Lemma~\ref{lem:core} for finite-valued constraint languages
was proved by the authors in~\cite{tz16:jacm}. Lemma~\ref{lem:core}, proved
below using Lemma~\ref{lem:killing}, has also been observed
in~\cite{Kozik15:icalp} and in~\cite{tz15:sidma}, where it was proved in
a different way (and without the use of Lemma~\ref{lem:killing}).

\begin{proof}
By definition, $\Gamma' = \Gamma[f(D)]$,
where $D$ is the domain of $\Gamma$ and $f \in \supp(\omega)$
for some unary fractional polymorphism $\omega$.
Assume that $I$ is satisfiable, and let $\sigma$ be an optimal assignment to $I$.
Now $f \circ \sigma$ is a satisfying assignment to $I'$, and by Lemma~\ref{lem:killing}, 
$f \circ \sigma$ is also an optimal assignment to $I$.
Conversely, any satisfying assignment to $I'$ is a satisfying assignment to $I$
of the same value.
\end{proof}

Let $\mathcal{C}_D = \{ \{(a)\} \mid a \in D \}$ be the set of constant unary
relations on the set $D$. It is known
(cf.~\cite[Proposition~20]{Kozik15:icalp}), that for a valued constraint
language $\Gamma$ on $D$ and a core $\Gamma'$ of $\Gamma$ on $D' \subseteq D$,
the problem $\VCSP(\Gamma' \cup \mathcal{C}_{D'})$ polynomial-time reduces to
$\VCSP(\Gamma$). In Theorem~\ref{thm:reductions}(\ref{red:coreconstants}) in
Section~\ref{sec:nec}, we present a stronger form of this reduction.

Let $\Gamma$ be a valued constraint language on $D$ with $\mathcal{C}_D\subseteq
\Gamma$. It is well known and easy to show that any $f\in\pol(\Gamma)$ is
idempotent~\cite{Bulatov05:classifying}.

\subsection{Relational Width}

We define \emph{relational width} which is the basis for our notion of valued
relational width.

Let $J$ be an instance of the CSP with $\phi_J(x_1,\dots,x_n) = \sum_{i=1}^q
\phi_i(\tup{x}_i)$, $X_i \subseteq V = \{x_1, \dots, x_n\}$ and $\phi_i
\colon D^{\ar(\phi_i)} \to \{0,\infty\}$.

For a tuple $\tup{t} \in D^{X}$, we denote by $\pi_{X'}(\tup{t})$ its projection onto $X' \subseteq X$.
For a constraint $\phi_i(\tup{x}_i)$, we define $\pi_{X'}(\phi_i) = \{
\pi_{X'}(\tup{t}) \mid \tup{t} \in \feas(\phi_i)\}$ where $X'\subseteq X_i$.

Let $1 \leq k \leq \ell$ be integers.
The following definition is equivalent\footnote{The two requirements
in~\cite{Barto16:jloc} are: for every $X\subseteq V$ with $|X|\leq\ell$ we have
$X\subseteq X_i$ for some $1\leq i\leq q$; and for every set $X\subseteq V$ with
$|X|\leq k$ and every $1\leq i,j\leq q$ with $X\subseteq X_i$ and $X\subseteq
X_j$ we have $\pi_{X}(\phi_i)=\pi_{X}(\phi_j)$.} to the definition of
$(k,\ell)$-minimality~\cite{Bulatov06:ja} for CSP instances
given in~\cite{Barto16:jloc}.

\begin{definition}
Let $J$ be an instance of the CSP with $\phi_J(x_1,\ldots,x_n)=\sum_{i=1}^q
\phi_i(\tup{x}_i)$, $X_i \subseteq V = \{x_1,\ldots,x_n\}$ and $\phi_i\colon
D^{\ar(\phi_i)}\to\{0,\infty\}$. Then $J$ is said to be \emph{$(k,\ell)$-minimal} if:
\begin{itemize}
\item
For every $X \subseteq V$, $\left|X\right| \leq \ell$, there exists $1 \leq i \leq q$ such that $X = X_i$.
\item
For every $i, j \in \left[q\right]$ such that $\left| X_j \right| \leq k$ and $X_j \subseteq X_i$,
$\phi_j = \pi_{X_j}(\phi_i)$.
\end{itemize}
\end{definition}

There is a straightforward polynomial-time algorithm for finding an equivalent $(k,\ell)$-minimal instance~\cite{Barto16:jloc}.
This leads to the notion of \emph{relational width}:

\begin{definition}
A constraint language $\Delta$ has relational width $(k,\ell)$ if, for every instance
$J$ of $\CSP(\Delta)$, an equivalent $(k,\ell)$-minimal instance is non-empty if,
and only if,
$J$ has a solution.
\end{definition}

An $m$-ary idempotent operation $f \colon D^m\to D$ is called a 
\emph{weak near-unanimity} (WNU) operation if, for all $x,y\in D$,
\begin{equation}\label{pseudownu}
f(y,x,x,\ldots,x)=f(x,y,x,x,\ldots,x)=f(x,x,\ldots,x,y)\,.
\end{equation}

\begin{definition}\label{def:bwc} 
We say that a clone of operations satisfies the \emph{bounded width condition (BWC)}
if it contains a (not necessarily idempotent) $m$-ary operation satisfying the identities (\ref{pseudownu}), for every $m\geq 3$.
\end{definition}

The following result is known as the ``bounded width theorem'' as it
characterises constraint languages of bounded relational width, that is,
constraint languages that are tractable via the $(k,\ell)$-minimality algorithm
for some $k\leq\ell$. 

\begin{theorem}[\hspace*{-0.3em}\cite{Barto14:jacm,Bulatov09:width,LaroseZadori07:au}]\label{thm:boundedrelwidth}
Let $\Delta$ be a constraint language of finite size containing all constant unary relations.
Then, $\Delta$ has bounded relational width if, and only if, $\pol(\Delta)$ satisfies the BWC.
\end{theorem}

Moreover, a collapse of relational width is known.

\begin{theorem}[\hspace*{-0.3em}\cite{Barto16:jloc,Bulatov09:width}]\label{thm:libor23}
Let $\Delta$ be a constraint language of finite size containing all constant unary relations.
If $\Delta$ has bounded relational width, then it has relational width $(2,3)$.
\end{theorem}

\begin{remark}
\label{remark:BWC}
We remark that most of the papers cited above use a different bounded width
condition, namely that of having WNU operations of all but
finitely many arities~\cite[Theorem~1.2]{Maroti08:weakly}.
By~\cite[Theorem~1.6~(4)]{Kozik15:au}, this is equivalent to
Definition~\ref{def:bwc}. 
Also note that our definition of the BWC does not require idempotency of the
operations. The reason is that we prove our main result, Theorem~\ref{thm:main}
below, without the requirement
of including the constant unary relations, which is often assumed in the
algebraic papers on the CSP.
\end{remark}

%
%
\section{The Power of Sherali-Adams Relaxations}
\label{sec:sapower}

In this section, we state our main result on the power of 
the Sherali-Adams linear programming relaxation~\cite{Sherali1990} to VCSPs.
We also give a number of applications of this result.
The Sherali-Adams linear programming relaxation is defined in Section~\ref{width}
and the characterisation of its power is stated in Section~\ref{sec:main}.
In Section~\ref{sec:algs}, we give a number of algorithmic consequences of our result
and Section~\ref{sec:compl} show how it can be used to derive complete complexity
classifications for large families of valued constraint languages.
In Section~\ref{sec:blp}, we compare our result to the characterisation of
valued relational width 1 which we obtained in previous work.
Finally, in Section~\ref{sec:meta}, we address the problem of finding an actual solution
and of determining whether or not
a valued constraint language has bounded valued relational width.

\subsection{Valued Relational Width}
\label{width}

Let $I$ be an instance of the VCSP with $\phi_\inst(x_1,\dots,x_n) = \sum_{i=1}^q
\phi_i(\tup{x}_i)$, $X_i \subseteq V = \{x_1, \dots, x_n\}$ and $\phi_i \colon
D^{\ar(\phi_i)} \to \qq$. A \emph{null constraint} is a constraint that has a
weighted relation identical to $0$. Ensure that for every non-empty $X
\subseteq V$ with $|X| \leq \ell$ there is some constraint $\phi_i(\tup{x}_i)$
with $X_i=X$, possibly by adding null constraints.

The Sherali-Adams relaxation with
parameters $(k,\ell)$, henceforth called the SA$(k,\ell)$-relaxation of $I$, is given by
the following linear program.
The variables are $\lambda_{i}(\sigma)$ for every $i \in \left[q\right]$ and assignment $\sigma \colon X_i \to D$.
We slightly abuse notation by writing $\sigma \in \feas(\phi_i)$ for $\sigma \colon X_i \to D$ such that $\sigma(\tup{x}_i) \in \feas(\phi_i)$.
\begin{align}
\min \sum_{i = 1}^q \sum_{\sigma \in \feas(\phi_i)} & \lambda_{i}(\sigma) \phi_i(\sigma(\tup{x}_i)) \nonumber\\
\lambda_{j}(\tau) &= \sum_{\substack{\sigma \colon X_i \to D\\ \Crestrict{\sigma}{X_j} = \tau}} \lambda_i(\sigma) &
   \forall i, j \in \left[q\right] : X_j \subseteq X_i, \left|X_j\right| \leq k, \tau \colon X_j \to D \label{sa:marginal} \\
\sum_{\sigma \colon X_i \to D} \lambda_{i}(\sigma) &= 1 &
   \forall i\in\left[q\right] \label{sa:sum1} \\
\lambda_{i}(\sigma) &= 0 &
   \forall i \in \left[q\right], \sigma \colon X_i \to D, \sigma(\tup{x}_i) \not\in \feas(\phi_i) \label{sa:infeas} \\
\lambda_{i}(\sigma) &\geq 0 &
   \forall i \in \left[q\right], \sigma \colon X_i \to D \label{sa:nonnegative}
\end{align}

The relaxation SA$(k,k)$ is often referred to as ``$k$ rounds of Sherali-Adams''.

We write $\lpval{I}{\lambda}{k}{\ell}$ for the value of the LP-solution $\lambda$ to the
SA$(k,\ell)$-relaxation of $I$, and $\lpopt{I}{k}{\ell}$ for its optimal value.

\begin{definition}
We say that a valued constraint language $\Gamma$ has \emph{valued relational width
$(k,\ell)$} if, for every instance $I$ of $\VCSP(\Gamma)$, 
$\vcspopt(\inst)=\lpopt{I}{k}{\ell}$ (i.e.\ the optimum of $I$ coincides with the optimum of the SA$(k,\ell)$-relaxation of
$ I$).
\end{definition}

When $\Gamma$ has valued relational width $(k,k)$ we also say that $\Gamma$ has valued relational width $k$.
When $\Gamma$ has valued relational width $k$ for some fixed $k \geq 1$, then we say that
$\Gamma$ has \emph{bounded valued relational width}.

We say that an instance $I$ of $\VCSP(\Gamma)$ is a \emph{gap instance for} SA$(k,\ell)$,
if its SA$(k,\ell)$ optimum is strictly smaller than its VCSP optimum.
Then, $\Gamma$ having bounded valued relational width is equivalent to saying that
there is some constant level of the Sherali-Adams hierarchy for which there are no gap instances
in $\VCSP(\Gamma)$.

\begin{definition}\label{def:reducesto}
Let $\Gamma$ and $\Delta$ be two valued constraint languages.
We write $\Delta \reducesto \Gamma$ if there is a polynomial-time reduction
from $\VCSP(\Delta)$ to $\VCSP(\Gamma)$ that preserves bounded valued relational
width.
\end{definition}

By Definition~\ref{def:reducesto}, $\reducesto$ reductions compose. 
Let $\Delta\reducesto\Gamma$. By Definition~\ref{def:reducesto}, (i) if $\Gamma$ has
bounded valued relational width then so does $\Delta$; (ii) if $\Delta$ does not
have bounded valued relational width then neither does $\Gamma$.

\subsection{A Characterisation of Bounded Valued Relational Width}
\label{sec:main}

The following characterisation of bounded valued relational width is our main result.
It precisely captures the power of Sherali-Adams relaxations for exact optimisation of VCSPs.

\begin{theorem}[Main]\label{thm:main}
Let $\Gamma$ be a valued constraint language of finite size. The following are equivalent:
\begin{enumerate}[(i)]
\item $\Gamma$ has bounded valued relational width. \label{cnd:bound}
\item $\Gamma$ has valued relational width $(2,3)$. \label{cnd:23}
\item $\supp(\Gamma)$ satisfies the BWC.  \label{cnd:BWC}
\end{enumerate}
\end{theorem}

The proof of Theorem~\ref{thm:main} is based on the following two theorems
which show that the BWC is a sufficient and necessary condition, respectively, 
for a constraint language to have bounded valued relational width.

\begin{theorem}\label{thm:mainicalp}
Let $\Gamma$ be a valued constraint language of finite size containing all constant unary relations.
If $\supp(\Gamma)$ satisfies the BWC, then $\Gamma$ has valued relational width $(2,3)$.
\end{theorem}

\begin{theorem}\label{thm:mainicalp2}
Let $\Gamma$ be a valued constraint language of finite size containing all constant unary relations.
If $\Gamma$ has bounded valued relational width, then
$\supp(\Gamma)$ satisfies the BWC.
\end{theorem}

We prove Theorems~\ref{thm:mainicalp} and~\ref{thm:mainicalp2} in Section~\ref{sec:suff}
and~\ref{sec:nec}, respectively.
In order to finish the proof of Theorem~\ref{thm:main}, we must reduce to the case when
the language $\Gamma$ is assumed to contain all constants.
This is done by taking a core $\Gamma'$ of $\Gamma$ on a domain $D' \subseteq D$
and adding $\mathcal{C}_{D'}$ to $\Gamma$'.
We need the following two lemmas to ensure that this can be carried out.
Lemma~\ref{lem:coreconstants} is proved in Section~\ref{sec:reductions}
(as Lemma~\ref{lem:core-preserve}).
Lemma~\ref{lem:corewnus} is proved in Section~\ref{sec:corewnus}.

\begin{lemma}\label{lem:coreconstants}
Let $\Gamma$ be a valued constraint language of finite size on domain $D$. If
$\Gamma'$ is a core of $\Gamma$ on domain $D' \subseteq D$, then $\Gamma' \cup
\mathcal{C}_{D'} \reducesto \Gamma$.
\end{lemma}

\begin{lemma}\label{lem:corewnus}
Let $\Gamma$ be a valued constraint language of finite size on domain $D$ and
$\Gamma'$ a core of $\Gamma$ on domain $D' \subseteq D$.
Then, $\supp(\Gamma)$ satisfies the BWC if, and only if, $\supp(\Gamma' \cup \mathcal{C}_{D'})$
satisfies the BWC.
\end{lemma}

\begin{proof}[Proof of Theorem~\ref{thm:main}]
The implication $(\ref{cnd:23})\Longrightarrow(\ref{cnd:bound})$ is trivial.
We first prove the implication $(\ref{cnd:BWC})\Longrightarrow(\ref{cnd:23})$.
Suppose that $\supp(\Gamma)$ satisfies the BWC. We start by going to a core of
$\Gamma$ and adding constant unary relations with the goal of applying
Theorem~\ref{thm:mainicalp}. Let $\Gamma'$ be a core of $\Gamma$ on domain $D'
\subseteq D$ and let $\Gamma'_c = \Gamma' \cup \mathcal{C}_{D'}$. By
Lemma~\ref{lem:corewnus}, $\supp(\Gamma'_c)$ also satisfies the BWC. By
Theorem~\ref{thm:mainicalp}, $\Gamma'_c$ has valued relational width $(2,3)$, so
clearly $\Gamma'$ has valued relational width $(2,3)$ as well. Every feasible
solution to the SA$(2,3)$-relaxation of an instance $I'$ of $\VCSP(\Gamma')$ is
also a feasible solution to the SA$(2,3)$-relaxation of the corresponding
instance $I$ of $\VCSP(\Gamma)$. The result now follows from Lemma~\ref{lem:core}
as the optimum of $I'$ and $I$ coincide.

It remains to prove the implication
$(\ref{cnd:bound})\Longrightarrow(\ref{cnd:BWC})$.
Suppose that $\supp(\Gamma)$ does not satisfy the BWC.
We start by going to a core of $\Gamma$ and adding constant unary relations with
the goal of applying Theorem~\ref{thm:mainicalp2}.
Let $\Gamma'$ be a core of $\Gamma$ on domain $D' \subseteq D$
and let $\Gamma'_c = \Gamma' \cup \mathcal{C}_{D'}$.
By Lemma~\ref{lem:corewnus}, $\supp(\Gamma'_c)$ does not satisfy the BWC.
By Theorem~\ref{thm:mainicalp2}, $\Gamma'_c$ does not have bounded valued
relational width. Finally, by Lemma~\ref{lem:coreconstants}, $\Gamma$ does not have
bounded valued relational width either.
\end{proof}

\subsection{Algorithmic Consequences}
\label{sec:algs}

We now give examples of previously studied valued constraint
languages and show that, as a corollary of Theorem~\ref{thm:main}, they all have, as
well as their generalisations, valued relational width $(2,3)$.

\begin{example}\label{ex:maj}
Let $\omega$ be a ternary fractional operation defined by
$\omega(f)=\omega(g)=\omega(h)=\frac{1}{3}$ for some 
(not necessarily distinct) majority operations $f$,
$g$, and $h$. Cohen et al. proved the tractability of any language improved by
$\omega$ by a reduction to CSPs with a majority
polymorphism~\cite{Cohen06:complexitysoft}.
\end{example}

\begin{example}\label{ex:cohen}
Let $\omega$ be a ternary fractional operation defined by $\omega(f)=\frac{2}{3}$ and $\omega(g)=\frac{1}{3}$, 
where $f:\{0,1\}^3\to\{0,1\}$ is the Boolean majority operation and
$g:\{0,1\}^3\to\{0,1\}$ is the Boolean minority operation.
Cohen et al. proved the tractability of any language improved by $\omega$ by a
simple propagation algorithm~\cite{Cohen06:complexitysoft}.
\end{example}

\begin{example}\label{ex:mjn}
Generalising Example~\ref{ex:cohen} from
Boolean to arbitrary domains, let $\omega$ be a ternary fractional operation
such that $\omega(f)=\frac{1}{3}$,
$\omega(g)=\frac{1}{3}$, and $\omega(h)=\frac{1}{3}$ for some (not necessarily
distinct) conservative majority operations $f$ and $g$, and a conservative minority operation $h$; such an
$\omega$ is called an MJN.
Kolmogorov and \v{Z}ivn\'y proved the tractability of any
language improved by $\omega$ by a 3-consistency algorithm and a reduction, via
Example~\ref{ex:stp}, to submodular function minimisation~\cite{kz13:jacm}.
\end{example}

The following corollary of Theorem~\ref{thm:main} generalises Examples~\ref{ex:maj}-\ref{ex:mjn}.

\begin{corollary}\label{cor:maj}
Let $\Gamma$ be a valued constraint language of finite size such that $\supp(\Gamma)$
contains a majority operation.
Then, $\Gamma$ has valued relational width $(2,3)$.
\end{corollary}
\begin{proof}
Let $f$ be a majority operation in $\supp(\Gamma)$.
Then, for every $k \geq 3$, $f$ generates a WNU $g_k$ of arity $k$:
$g_k(x_1,\dots,x_k) = f(x_1,x_2,x_3)$.
By Lemma~\ref{lem:suppclone}, $\supp(\Gamma)$ is a clone,
so $g_k \in \supp(\Gamma)$ for all $k \geq 3$.
Therefore, $\supp(\Gamma)$ satisfies the BWC and
the result follows from Theorem~\ref{thm:main}.
\end{proof}

\begin{example}\label{ex:stp}
Let $\omega$ be a binary fractional operation defined by
$\omega(f)=\omega(g)=\frac{1}{2}$, where $f$ and $g$ are conservative and
commutative operations and $f(x,y)\neq g(x,y)$ for every $x$ and $y$; such an
$\omega$ is called a \emph{symmetric tournament pair} (STP). Cohen et al. proved
the tractability of any language improved by $\omega$ by a 3-consistency
algorithm and an ingenious reduction to submodular function
minimisation~\cite{Cohen08:Generalising}. Such languages were shown to be the
only tractable languages among conservative finite-valued constraint
languages~\cite{kz13:jacm}. 
\end{example}

The following corollary of Theorem~\ref{thm:main} generalises Example~\ref{ex:stp}.

\begin{corollary}\label{cor:stp}
Let $\Gamma$ be a valued constraint language of finite size such that $\supp(\Gamma)$
contains two symmetric tournament operations (that is, binary operations
$f$ and $g$ that are both conservative and commutative and $f(x,y)\neq g(x,y)$
for every $x$ and $y$).
Then, $\Gamma$ has valued relational width $(2,3)$.
\end{corollary}

\begin{proof}
It is straightforward to verify that $h(x,y,z)=f(f(g(x,y),g(x,z)),g(y,z))$ is a
majority operation, as observed in~\cite[Corollary~5.8]{Cohen08:Generalising}.
The claim then follows from Corollary~\ref{cor:maj}.
\end{proof}

\begin{example}\label{ex:tp}
Generalising Example~\ref{ex:stp},
let $\omega$ be a binary fractional operation defined by
$\omega(f)=\omega(g)=\frac{1}{2}$, where $f$ and $g$ are conservative and
commutative operations; such
an $\omega$ is called a \emph{tournament pair}. Cohen et al. proved the
tractability of any language improved by $\omega$ by a consistency-reduction
relying on Bulatov's result~\cite{Bulatov06:ja}, which in turn relies on
3-consistency, to the STP case from Example~\ref{ex:stp}~\cite{Cohen08:Generalising}.
\end{example}

The following corollary of Theorem~\ref{thm:main} generalises Example~\ref{ex:tp}.

\begin{corollary}\label{cor:tp}
Let $\Gamma$ be a valued constraint language of finite size such that $\supp(\Gamma)$
contains a tournament operation (that is, a binary conservative and commutative
operation).
Then, $\Gamma$ has valued relational width $(2,3)$.
\end{corollary}

\begin{proof} 
Let $f$ be a tournament operation from $\supp(\Gamma)$. We claim that $f$ is a
2-semilattice; that is, $f$ is idempotent, commutative, and satisfies the
restricted associativity law $f(x,f(x,y)) = f(f(x,x),y)$.  To see that, notice
that $f(x,f(x,y))=x$ if $f(x,y)=x$ and $f(x,f(x,y))=y$ if $f(x,y)=y$; together,
$f(x,f(x,y))=f(x,y)$. On the other hand, trivially $f(f(x,x),y)=f(x,y)$. 

Also note that $f(x,f(y,x))=f(x,f(x,y))=f(x,y)$. 
For every $k \geq 3$, $f$ generates a WNU $g_k$ of arity $k$:
$g_k(x_1,\dots,x_k) = f(f(\ldots(f(x_1,x_2),x_3),\ldots),x_k)$.
By Lemma~\ref{lem:suppclone}, $\supp(\Gamma)$ is a clone,
so $g_k \in \supp(\Gamma)$ for all $k \geq 3$.
Therefore, $\supp(\Gamma)$ satisfies the BWC so
the result follows from Theorem~\ref{thm:main}.
\end{proof}

\begin{example}\label{ex:stp-mjn}
In this example we denote by $\multiset{\ldots}$ a multiset.
Let $\omega$ be a binary fractional operation on $D$ defined by
$\omega(f)=\omega(g)=\frac{1}{2}$ and let $\mu$ be a ternary fractional
operation on $D$ defined by $\mu(h_1)=\mu(h_2)=\mu(h_3)=\frac{1}{3}$. 
Suppose that $\multiset{f(x,y),g(x,y)}=\multiset{x,y}$ for every $x$ and $y$ and
$\multiset{h_1(x,y,z),h_2(x,y,z),h_3(x,y,z)}=\multiset{x,y,z}$ for every $x$, $y$, and
$z$. 
Moreover, suppose that for every two-element subset $\{a,b\} \subseteq D$,
either $\omega|_{\{a,b\}}$ is an STP or $\mu|_{\{a,b\}}$ is an MJN.
Let $\Gamma$ be a language on $D$ improved by a fractional polymorphism $\omega$
as described above. 
Kolmogorov and \v{Z}ivn\'y proved the tractability of $\Gamma$ by a
3-consistency algorithm and a reduction, via Example~\ref{ex:stp}, to submodular
function minimisation~\cite{kz13:jacm}. Such languages were shown to be the only
tractable languages among conservative valued constraint
languages~\cite{kz13:jacm}.
We will discuss conservative valued constraint languages in more detail in
Section~\ref{sec:compl}.
\end{example}

The following corollary of Theorem~\ref{thm:main} covers Example~\ref{ex:stp-mjn}.

\begin{corollary}\label{cor:stp-mjn}
Let $\Gamma$ be a valued constraint language of finite size with fractional polymorphisms
$\omega$ and $\mu$ as described in Example~\ref{ex:stp-mjn}.
Then, $\Gamma$ has valued relational width $(2,3)$.
\end{corollary}

\begin{proof}
Let $P$ be the set of $2$-element subsets of $D$ such that $\omega|_{\{a,b\}}$
is an STP for $\{a,b\}\in P$ and $\mu|_{\{a,b\}}$ is an MJN for $\{a,b\}\not\in
P$. 
Let $p(x,y,z)=f(f(g(y,x),g(x,z)),g(y,z))$. Observe that $p|_{\{a,b\}}$ is a
majority for $\{a,b\}\in P$, and $p|_{\{a,b\}}$ is either $\proj^{(3)}_1$ or
$\proj^{(3)}_2$ for $\{a,b\}\not\in P$ (possibly different projections for
different $2$-element subsets from $P$). Now let
$q(x,y,z)=p(h_1(x,y,z),h_2(x,y,z),h_3(x,y,z))$. 
For $x,y\in\{a,b\}\in P$, $q(x,x,y)=q(x,y,x)=q(y,x,x)=p(\multiset{x,x,y})=x$.
For $x,y\in\{a,b\}\not\in P$, $q(x,x,y)=q(x,y,x)=q(y,x,x)=p(x,x,y)=x$
as $p$ is either the first or the second projection. Thus, $q$ is
a majority operation. 
The claim then follows from Corollary~\ref{cor:maj}.
\end{proof}

\subsection{Complexity Consequences}
\label{sec:compl}

We now give some computational complexity consequences of
Theorem~\ref{thm:main}. First, we obtain a new and simpler proof (in fact two
proofs) of the complexity classification of conservative valued constraint
languages~\cite{kz13:jacm}. Second, we obtain a complexity classification of
(generalisation of) Minimum-Solution problems over arbitrary finite domains.

\emph{Minimum-Solution} (Min-Sol) problems~\cite{Jonsson08:max-sol}, studied
under the name of Min-Ones on Boolean
domains~\cite{Creignouetal:siam01,Khanna00:approximability}, constitute a large
and interesting subclass of VCSPs including, for instance, 
integer linear programming over bounded domains.

\begin{definition}
A valued constraint language $\Gamma$ on finite domain $D$ is called a
\emph{Min-Sol language} if
$\Gamma=\Delta\cup\{\nu\}$, where $\Delta$ is a crisp constraint language on $D$ 
and $\nu:D\to\mathbb{Q}$ is an injective
finite-valued weighted relation.
\end{definition}

In other words, in Min-Sol problems the optimisation part of the objective
function is a sum of unary terms involving an injective finite-valued weighted
relation. 

As our main result in this section, we give a complexity classification of
\emph{all} Min-Sol languages on \emph{arbitrary} finite domains, thus improving
on previous classifications obtained for Min-Sol languages on domains with two
elements~\cite{Khanna00:approximability}, three elements~\cite{Uppman13:icalp},
and other special
cases~\cite{Jonsson07:maxsol,Jonsson08:max-sol,Jonsson08:siam}.

By 
Lemma~\ref{lem:coreconstants},
we can, without loss
of generality, restrict our attention to languages that include constants.

\begin{theorem}\label{thm:min-sol}
Let $D$ be an arbitrary finite domain and let $\Gamma=\Delta \cup \{\nu\}$ be an
arbitrary Min-Sol language of finite size on $D$ with $\mathcal{C}_D \subseteq \Gamma$.
Then, either $\supp(\Gamma)$ satisfies the BWC, in which case $\Gamma$ has
valued relational width $(2,3)$, or $\VCSP(\Gamma)$ is NP-hard. 
\end{theorem}

In order to prove Theorem~\ref{thm:min-sol}, we prove a more general result
classifying valued constraint languages that can express an injective unary
finite-valued weighted relation. Theorem~\ref{thm:min-sol} is then a simple
corollary of the following result.

\begin{theorem}\label{thm:injective}
Let $D$ be an arbitrary finite domain and let $\Gamma$ be an arbitrary valued constraint
language of finite size on $D$ with $\mathcal{C}_D\subseteq \Gamma$. Assume that
$\Gamma$ expresses a unary finite-valued weighted relation $\nu$ that is injective on $D$.
Then, either $\supp(\Gamma)$ satisfies the BWC, in which case $\Gamma$ has
valued relational width $(2,3)$, or $\VCSP(\Gamma)$ is NP-hard. 
\end{theorem}

We now define conservative valued constraint languages~\cite{kz13:jacm}.

\begin{definition}
A valued constraint language $\Gamma$ on $D$ is called \emph{conservative} if
$\Gamma$ contains all $\{0,1\}$-valued unary weighted relations. 
\end{definition}

We remark that for crisp constraint languages a different definition is
used~\cite{Bulatov11:conservative}.

Note that any
conservative language $\Gamma$ is a core and by Lemma~\ref{lem:coreconstants} we
can assume that $\mathcal{C}_D\subseteq \Gamma$.

Theorem~\ref{thm:injective} implies the following dichotomy theorem, first
established in~\cite{kz13:jacm} with the help of~\cite{Takhanov10:stacs}.

\begin{theorem}
Let $D$ be an arbitrary finite domain and let $\Gamma$ be an arbitrary 
conservative valued constraint language on $D$.
Then, either $\supp(\Gamma)$ satisfies the BWC, in which case
$\Gamma$ has valued relational width $(2,3)$, or $\VCSP(\Gamma)$ is NP-hard. 
\end{theorem}

We now give a different proof classifying conservative valued constraint
languages that relies on~\cite{Takhanov10:stacs} but has the advantage of giving
a more specific tractability criterion than the BWC that is different from the
STP/MJN criterion established in~\cite{kz13:jacm} and discussed in
Example~\ref{ex:stp-mjn}.

The following theorem was proved by Takhanov~\cite{Takhanov10:stacs} with a 
small strengthening in~\cite{kz13:jacm}.
%

\begin{theorem}[\hspace*{-0.3em}\cite{kz13:jacm,Takhanov10:stacs}]\label{thm:takhanov-kz}
Let $\Gamma$ be a conservative valued constraint language.
If $\pol(\Gamma)$ does not contain a majority polymorphism, then $\VCSP(\Gamma)$ is NP-hard.
\end{theorem}

We can strengthen Theorem~\ref{thm:takhanov-kz} to show NP-hardness of
$\VCSP(\Gamma)$ for a conservative valued constraint language $\Gamma$ which 
lacks a majority operation in the support clone of $\Gamma$. Consequently, we
obtain an alternative tractability criterion for conservative valued constraint
languages to the original criterion~\cite{kz13:jacm} that involved a binary
STP multimorphism and a ternary MJN multimorphism (cf.
Example~\ref{ex:stp-mjn}).

\begin{theorem}\label{thm:cons2}
Let $\Gamma$ be a conservative valued constraint language.
Either $\VCSP(\Gamma)$ is NP-hard, or $\supp(\Gamma)$ contains a majority operation
and hence $\Gamma$ has valued relational width $(2,3)$.
\end{theorem}

\subsection{Related Work on BLP and Relational Width}
\label{sec:blp}

The SA$(1,1)$ relaxation is also known as the \emph{basic linear programming relaxation} (BLP).
The following result capturing the power of BLP has been
established.\footnote{Theorem~\ref{thm:blp} as stated here follows
from~\cite[Corollary~3]{ktz15:sicomp} using Lemma~\ref{lem:suppclone}.}

An $m$-ary operation $f:D^m\to D$ is called \emph{symmetric} if
$f(x_1,\ldots,x_m)=f(x_{\pi(1)},\ldots,x_{\pi(m)})$ for every permutation $\pi$
of $\{1,\ldots,m\}$.

\begin{definition} 
We say that a clone of operation satisfies the \emph{SYM} condition if it
contains an $m$-ary symmetric operation, for every $m\geq 2$.
\end{definition}

\begin{theorem}[\hspace*{-0.3em}\cite{ktz15:sicomp}]\label{thm:blp}
Let $\Gamma$ be a valued constraint language of finite size. Then the following are equivalent:
\begin{enumerate}
\item $\Gamma$ has valued relational width $1$.
\item $\supp(\Gamma)$ satisfies the SYM.
\end{enumerate}
\end{theorem}

By definition, the SA$(1,\ell)$-relaxation is at least as tight as the
SA$(1,1)$-relaxation; i.e., any solution to the SA$(1,\ell)$-relaxation gives a
solution to the SA$(1,1)$-relaxation of the same value.
Hence any language with valued relational width $1$ has
valued relational width $(1,\ell)$. We now show that for any fixed $\ell$,
SA$(1,1)$ and SA$(1,\ell)$ have the same power. 

\begin{proposition}\label{prop:width1}
Let $\Gamma$ be a valued constraint language of finite size and let $\ell>1$ be fixed. If
$\Gamma$ has valued relational width $(1,\ell)$ then $\Gamma$ has valued
relational width $1$.
\end{proposition}

\begin{proof}
Let $\inst$ be an instance of $\VCSP(\Gamma)$.
Assume that $\vcspopt(\inst)=\lpopt{\inst}{1}{\ell}$ for the SA$(1,\ell)$-relaxation
of $\inst$. For the sake of contradiction, suppose that
$\vcspopt(\inst)>\lpopt{\inst}{1}{1}$
for the SA$(1,1)$-relaxation of $\inst$ and let $\lambda^*$ be an optimal
solution to SA$(1,1)$ of value $OPT^*$. Define $\lambda'$ as
follows. If $\lambda_i(\sigma)$ is a variable of SA$(1,1)$
then $\lambda'_i(\sigma)=\lambda_i^*(\sigma)$. Otherwise, let $\lambda_i(\sigma)$
correspond to the $i$th valued constraint $\phi_i(\tup{x}_i)$ with variables
$\{x_1,\ldots,x_r\}$.  We define $\lambda'_i(\sigma)$ as the product of the
$\lambda^*$'s corresponding to $\sigma(x_j)$, $1\leq j\leq r$. More formally, if $\phi_j(x_j)$
are the valued constraints with the scope $x_j$, for $1\leq j\leq r$, then we
define $\lambda'_i(\sigma)=\prod_{j=1}^r\lambda_j^*(\sigma(x_j))$. By the
definition of $\lambda'$, $\lambda'$ is a feasible solution to SA$(1,\ell)$. By
the definition of the SA relaxations, the extra valued constraints present in SA$(1,\ell)$ but
missing in SA$(1,1)$ are null and thus 
$\lpval{I}{\lambda'}{1}{\ell}=OPT^*<\vcspopt(\inst)$. But this contradicts $\Gamma$ having valued
relational width $(1,\ell)$.
\end{proof}

\begin{corollary}
Let $\Gamma$ be a valued constraint language of finite size. Then, the valued relational width
of $\Gamma$ is either $1$, or $2$, or $(2,3)$, or unbounded.
\end{corollary}

\begin{proof} 
If the valued relational width of $\Gamma$ is bounded then it is $(2,3)$, by
Theorem~\ref{thm:main}. If the valued relational width of $\Gamma$ is $(1,\ell)$
for some $\ell>1$ then it is $1$, by Proposition~\ref{prop:width1}. 
\end{proof}

There are valued constraint languages that have valued relational width $(2,3)$
but not $1$. For example, languages improved by a tournament pair fractional
polymorphism~\cite{Cohen08:Generalising}, discussed in detail in
Example~\ref{ex:tp} in Section~\ref{sec:algs}, have valued relational width
$(2,3)$ by the results in this paper, but do not have valued relational width in
$1$ as shown~\cite[Example~5]{ktz15:sicomp} using Theorem~\ref{thm:blp}.

It could be that either SA$(1)$ and SA$(2)$, or SA$(2)$ and SA$(2,3)$ have the
same power. The former happens in case of relational width. Dalmau
proved that if a crisp language has relational width $2$ then it has
relational width $1$~\cite{Dalmau09:ipl}. Together with
Theorem~\ref{thm:libor23} and the analogue of Proposition~\ref{prop:width1} for
relational width established in~\cite{Feder98:monotone},
this gives a trichotomy for relational width.

\begin{theorem}[\hspace*{-0.3em}\cite{Feder98:monotone,Dalmau09:ipl,Barto16:jloc}]\label{thm:trichotomy}
Let $\Delta$ be a crisp constraint language of finite size. Then precisely one of the following is true:
\begin{enumerate}
\item $\Delta$ has relational width $1$.
\item $\Delta$ has relational width $(2,3)$ and does not have relational width
$2$, nor $(1,\ell)$ for any $\ell\geq 1$.
\item $\Delta$ does not have bounded relational width.
\end{enumerate}
\end{theorem}

\begin{remark} It follows from the definitions that if a crisp constraint language
$\Delta$ has relational width $(k,\ell)$ then
$\Delta$ also has valued relational width $(k,\ell)$. However, the converse does
not hold. There exists a constraint language on a three-element domain with two
relations that has valued relational width $1$ but not relational width
$1$\cite[Example~99]{Kun16:clones}. 
\end{remark}

\subsection{Obtaining a Solution and the Meta Problem}
\label{sec:meta}

We now address two questions related to our main result.

Firstly, we show that for any VCSP instance over a language of valued relational
width $(2,3)$ we can not only compute the value of an optimal solution but we
can also find an optimal assignment in polynomial time.

\begin{proposition}\label{prop:self-reduce}
Let $\Gamma$ be a valued constraint language of finite size and $I$ an instance of $\VCSP(\Gamma)$.
If $\supp(\Gamma)$ satisfies the BWC, then an optimal assignment to $I$ can be found in polynomial time.
\end{proposition}

\begin{proof}
Let $\Gamma'$ be a core of $\Gamma$ on domain $D'$, 
and let $\Gamma_c =  \Gamma' \cup \{ \mathcal{C}_{D'} \}$.
By Lemma~\ref{lem:corewnus}, $\supp(\Gamma_c)$ satisfies the BWC,
so by Theorem~\ref{thm:mainicalp} we can obtain the optimum of $I$
by solving a linear programming relaxation.
Now, we can use self-reduction to obtain an optimal assignment.
It suffices to modify the instance $I$ to successively force each variable
to take on each value of $D'$. Whenever the optimum of the modified instance
matches that of the original instance, we can move on to assign the next variable.
This means that we need to solve at most $1+\left|V\right| \left|D'\right|$ linear
programming relaxations before finding an optimal assignment,
where $V$ is the set of variables of $I$.
\end{proof}

Secondly, we show that testing for the BWC is a decidable problem.
We rely on the following result that was proved in~\cite{Kozik15:au}, and also
follows from results in~\cite{Barto16:jloc}.

\begin{theorem}[\hspace*{-0.3em}\cite{Kozik15:au}]\label{thm:34}
An idempotent clone of operations satisfies the BWC if, and only if, it
contains a ternary WNU $f$ and a quaternary WNU $g$ with $f(y,x,x)=g(y,x,x,x)$ for
all $x$ and $y$.
\end{theorem}

\begin{proposition}
Testing whether a valued constraint language of finite size
satisfies the BWC is decidable.
\end{proposition}

\begin{proof}
Let $\Gamma$ be a valued constraint language of finite size on domain $D$. Let
$\Gamma'$ be a core of $\Gamma$ defined on domain $D'\subseteq D$. Finding $D'$
and $\Gamma'$ can be done via linear
programming~\cite[Section~4]{tz16:jacm}. 
By Lemma~\ref{lem:corewnus}, $\supp(\Gamma)$
satisfies the BWC if, and only if, $\supp(\Gamma'\cup \mathcal{C}_{D'})$ satisfies the
BWC. As constant unary relations enforce idempotency,
by Theorem~\ref{thm:34},
$\supp(\Gamma'\cup\mathcal{C}_{D'})$ satisfies the BWC if, and only if,
$\supp(\Gamma'\cup\mathcal{C}_{D'})$ contains a ternary WNU $f$ and a 4-ary 
WNU $g$ with $f(y,x,x)=g(y,x,x,x)$ for all $x$ and $y$. It is easy to write a
linear program that checks for this condition, as it has been
done in the context of finite-valued constraint
languages~\cite[Section~4]{tz16:jacm}.
\end{proof}

\section{Sufficiency: Proof of Theorem~\ref{thm:mainicalp}}
\label{sec:suff}

In this section, we prove that the BWC is a sufficient condition for a valued
constraint language with all constant unary relations to have valued relational width $(2,3)$.

We start with a technical lemma.
For a feasible solution $\lambda$ of SA$(k,\ell)$, let
$\supp(\lambda_i) = \{ \sigma \colon X_i \to D \mid \lambda_i(\sigma) > 0 \}$.

\begin{lemma}\label{lem:fullsupport}
Let $I$ be an instance of $\VCSP(\Gamma)$.
Assume that SA$(k,\ell)$ for $I$ is feasible.
Then, there exists an optimal solution $\lambda^*$ to SA$(k,\ell)$ such that,
for every $i$,
$\supp(\lambda^*_i)$ is closed under 
every operation in $\supp(\Gamma)$.
\end{lemma}

\begin{proof}
Let $\omega$ be an arbitrary $m$-ary fractional polymorphism of $\Gamma$,
and let $\lambda$ be any feasible solution $\lambda$ to SA$(k,\ell)$. Define
$\lambda^\omega$ by
\[
\lambda^\omega_i(\sigma) = \Pr_{\substack{f \sim \omega \\ \sigma_1,\dots,\sigma_m \sim \lambda_i}}
[ f \circ (\sigma_1,\dots,\sigma_m) = \sigma].
\]
We show that $\lambda^\omega$ is a feasible solution to SA$(k,\ell)$, and that if $\lambda$ is optimal,
then so is $\lambda^\omega$.

Clearly $\lambda^\omega_i$ is a probability distribution for each $i \in \left[q\right]$,
so (\ref{sa:sum1}) and (\ref{sa:nonnegative}) hold.
Since $\omega$ is a fractional polymorphism of $\Gamma$,
we have $\sigma \in \feas(\phi_i)$ 
for any choice of $f \in \supp(\omega)$ and 
$\sigma_1, \dots, \sigma_m \in \supp(\lambda_i)$.
Hence, $\lambda^\omega_i(\sigma) = 0$ for $\sigma \not\in \feas(\phi_i)$,
so (\ref{sa:infeas}) holds.

Finally, let $j \in \left[q\right]$ be such that $X_j \subseteq X_i$, $\left|X_j\right| \leq k$,
and let $\tau \colon X_j \to D$. Then,
\begin{align*}
\sum_{\sigma \colon X_i \to D, \Crestrict{\sigma}{X_j} = \tau} \lambda^\omega_i(\sigma)
&=
\sum_{\sigma \colon X_i \to D, \Crestrict{\sigma}{X_j} = \tau} \Pr_{\substack{f \sim \omega \\ \sigma_1, \dots, \sigma_m \sim \lambda_i}}
[ f \circ (\sigma_1,\dots,\sigma_m) = \sigma]\\
&=
\Pr_{\substack{f \sim \omega \\ \sigma_1, \dots, \sigma_m \sim \lambda_i}}
[ \Crestrict{(f \circ (\sigma_1,\dots,\sigma_m))}{X_j} = \tau]\\
&=
\Pr_{\substack{f \sim \omega \\ \sigma_1, \dots, \sigma_m \sim \lambda_i}}
[ f \circ (\sigma_1|_{X_j},\dots,\sigma_m|_{X_j}) = \tau]\\
&=
\sum_{\tau_1, \dots, \tau_m \colon X_j \to D}
\Pr_{\substack{f \sim \omega \\ \sigma_1,\dots,\sigma_m \sim \lambda_i}}
[ \Crestrict{\sigma_1}{X_j} = \tau_1, \dots, \Crestrict{\sigma_m}{X_j} = \tau_m,
f \circ (\tau_1,\dots,\tau_m) = \tau]\\
&=
\sum_{\tau_1, \dots, \tau_m \colon X_j \to D}
\lambda_j(\tau_1) \cdots \lambda_j(\tau_m)
\Pr_{f \sim \omega}
[ f \circ (\tau_1,\dots,\tau_m) = \tau]\\
&=
\Pr_{\substack{f \sim \omega \\ \tau_1, \dots, \tau_m \sim \lambda_j}}
[ f \circ (\tau_1,\dots,\tau_m) = \tau]\\
&=
\lambda^\omega_j(\tau),
\end{align*}
where, we have used the fact that
(\ref{sa:marginal}) can be read as 
$\lambda_j(\tau) = \Pr_{\sigma \sim \lambda_i} \left[\Crestrict{\sigma}{X_j} =
\tau\right]$.
It follows that (\ref{sa:marginal}) also holds for $\lambda^\omega$, so $\lambda^\omega$ is feasible.

For each $i \in \left[q\right]$, we have:
\begin{align*}
\sum_{\sigma \in \feas(\phi_i)} \lambda_i(\sigma)\phi_i(\sigma(\tup{x}_i)) &=
\E_{\sigma \sim \lambda_i} \phi_i(\sigma) = \E_{\sigma_1,\dots,\sigma_m \sim
\lambda_i} \frac{1}{m} \sum_{j=1}^m \phi_i(\sigma_j(\tup{x}_i))\\
&\geq \E_{\substack{f \sim \omega\\ \sigma_1,\dots,\sigma_m \sim \lambda_i}} \phi_i(f(\sigma_1(\tup{x}_i),\dots,\sigma_m(\tup{x}_i)))\\
&= \sum_{\sigma \in \feas(\phi_i)} \Big( \Pr_{\substack{f \sim \omega\\ \sigma_1,\dots,\sigma_m \sim \lambda_i}} \left[f \circ (\sigma_1,\dots,\sigma_m) = \sigma\right] \Big) \phi_i(\sigma(\tup{x}_i))\\
&= \sum_{\sigma \in \feas(\phi_i)} \lambda^\omega_i(\sigma) \phi_i(\sigma(\tup{x}_i)).
\end{align*}
Therefore, if $\lambda$ is optimal, then $\lambda^\omega$ must also be optimal.

Now assume that $\lambda$ is an optimal solution and that $\supp(\lambda)$ is not
closed under some operation $f \in \supp(\omega)$ for $\omega \in \fpol(\Gamma)$,
i.e.\ for some $\sigma_1, \dots, \sigma_m \in \supp(\lambda)$, we have
$f(\sigma_1,\dots,\sigma_m) \not\in \supp(\lambda)$.
But note that  $f(\sigma_1,\dots,\sigma_m) \in \supp(\lambda^\omega_i)$.
Therefore, $\lambda' = \frac{1}{2}(\lambda+\lambda^\omega)$ is an optimal solution
such that $\supp(\lambda_i) \subsetneq \supp(\lambda'_i) \subseteq
D^{X_i}$.
For each $i \in \left[q\right]$, $D^{X_i}$ is finite.
Hence, by repeating this procedure, we obtain a sequence of optimal solutions with
strictly increasing support until, after a finite number of steps, we obtain a $\lambda^*$
that is closed under
every operation in $\supp(\Gamma)$.
\end{proof}

We now have everything that is needed to prove Theorem~\ref{thm:mainicalp}.

\begin{proof}[Proof of Theorem~\ref{thm:mainicalp}]
Let $I$ be an instance of $\VCSP(\Gamma)$ with $\phi_\inst(x_1,\dots,x_n) = \sum_{i=1}^q
\phi_i(\tup{x}_i)$, $X_i \subseteq V = \{x_1, \dots, x_n\}$ and $\phi_i \colon
D^{\ar(\phi_i)} \to \qq$.

The dual of the SA$(k,\ell)$ relaxation can be written in the following form,
with variables $z_i$ for $i \in \left[q\right]$ and $y_{j, \tau, i}$ for $i, j \in \left[q\right]$ such that $X_j \subseteq X_i$, $\left|X_j\right| \leq k$, and $\tau \colon X_j \to D$.
The dual variables corresponding to $\lambda_i(\sigma) = 0$ are eliminated together with
the dual inequalities for $i, \sigma \not\in \feas(\phi_i)$.

\begin{align}
\max \sum_{i=1}^q z_i \nonumber\\
z_i &\leq \phi_i(\sigma) + \sum_{j \in \left[q\right], X_j \subseteq X_i} y_{j, \Crestrict{\sigma}{X_j}, i} - \sum_{j \in \left[q\right], X_i \subseteq X_{j}} y_{i, \sigma, j}  & \forall i \in \left[q\right], \left|X_i\right| \leq k, \sigma \in \feas(\phi_i) \label{dual:smallk} \\
z_i &\leq \phi_i(\sigma) + \sum_{\substack{j \in \left[q\right], X_j \subseteq X_i \\ \left|X_j\right| \leq k}} y_{j,\Crestrict{\sigma}{X_j},i} & \forall i \in \left[q\right], |X_i| > k, \sigma \in \feas(\phi_i) \label{dual:largek}
\end{align}

It is clear that if $I$ has a feasible solution, then so does the SA$(k,\ell)$ primal.
Assume that the SA$(2,3)$-relaxation has a feasible solution.

By Lemma~\ref{lem:fullsupport}, there exists an optimal primal solution
$\lambda^*$ such that, for every $i \in \left[q\right]$, $\supp(\lambda^*_i)$ is
closed under $\supp(\Gamma)$. Let $y^*$, $z^*$ be an optimal dual solution.

Let $\Delta = \{\phi'_i\}_{i=1}^q \cup \{ \mathcal{C}_D \}$, where $\phi'_i =
\supp(\lambda^*_i)$, i.e. $\phi'_i(\tup{x})=0$ if $\tup{x}\in\supp(\lambda^*_i)$
and $\phi'_i(\tup{x})=\infty$ otherwise. We consider the instance $J$ of
CSP$(\Delta)$ with $\phi_J(x_1,\ldots,x_n)=\sum_{i=1}^q\phi'_i(\tup{x}_i)$.

We make the following observations:
\begin{enumerate}
\item
By construction of $\lambda^*$, $\supp(\Gamma) \subseteq \pol(\Delta)$, so $\Delta$
contains all constant unary relations and
satisfies the BWC.
By Theorems~\ref{thm:boundedrelwidth} and~\ref{thm:libor23}, the language $\Delta$ has relational width $(2,3)$. 
\item
The first set of constraints in the primal say that if $i, j \in \left[q\right]$,
$\left|X_j\right| \leq 2$ and $X_j \subseteq X_i$, then $\lambda^*_j(\tau) >
0$ (i.e.\ $\tau \in \phi'_j$) if, and only if, $\sum_{\sigma \colon X_i \to D,
\Crestrict{\sigma}{X_j} = \tau} \lambda^*_i(\sigma) > 0$ (i.e.\ $\tau$ satisfies
$\pi_{X_j}(\phi'_i)$). In other words, $J$ is $(2,3)$-minimal.
\end{enumerate}

These two observations imply that $J$ has a satisfying assignment $\alpha \colon V \to D$.
Let $\alpha_i = \Crestrict{\alpha}{X_i}$.
By complementary slackness, since $\lambda^*_i(\alpha_i) > 0$ for every $i \in \left[q\right]$, 
we must have equality in the corresponding rows in the dual indexed by $i$ and $\alpha_i$.
We sum these rows over $i$:
\begin{equation}\label{eq:dualsum}
\sum_{i=1}^q z^*_i\ =\ \sum_{i=1}^q \phi_i(\alpha(\tup{x}_i)) +
\Big( \sum_{i=1}^q \sum_{\substack{j \in \left[q\right], X_j \subseteq X_i \\ \left|X_j\right| \leq 2}} y^*_{j,\Crestrict{\alpha_i}{X_j},i} -
 \sum_{\substack{i \in \left[q\right]\\ \left|X_i\right| \leq 2}}
 \sum_{\substack{j\in\left[q\right]\\ X_i \subseteq X_j}} y^*_{i,\alpha_i,j}
\Big)\,.
\end{equation}

By noting that $\Crestrict{\alpha_i}{X_j} = \alpha_j$ when $X_j \subseteq X_i$, 
we can rewrite the expression in parenthesis on
the right-hand side of (\ref{eq:dualsum}) as:
\begin{equation}\label{eq:is0}
 \sum_{\substack{i, j \in \left[q\right], X_j \subseteq X_i\\ \left|X_j\right| \leq 2}} y^*_{j, \alpha_j, i} -
 \sum_{\substack{i, j \in \left[q\right], X_i \subseteq X_j\\ \left|X_i\right| \leq 2}} y^*_{i, \alpha_i, j}\ =\ 0.
\end{equation}

Therefore,
\[
\sum_{i=1}^q \sum_{\sigma \in \feas(\phi_i)} \lambda^*_i(\sigma) \phi_i(\sigma(\tup{x}_i)) = 
\sum_{i=1}^q z^*_i\ =\ \sum_{i=1}^q \phi_i(\alpha(\tup{x}_i)),
\]
where the first equality follows by strong LP-duality, and the second by (\ref{eq:dualsum}) and
(\ref{eq:is0}).
Since $I$ was an arbitrary instance of $\VCSP(\Gamma)$,
the theorem follows.
\end{proof}

%
%
\section{Necessity: Proof of Theorem~\ref{thm:mainicalp2}}
\label{sec:nec}

In this section, we prove that the BWC is a necessary condition for a valued
constraint language with all constant unary relations 
to have bounded valued relational width.

The main idea of the proof is to show that if $\supp(\Gamma)$ does not satisfy the BWC,
then $\Gamma$ can, in a sense, simulate linear equations in some Abelian group.
We show that such linear equations do not have bounded valued relational width,
and that the simulation preserves bounded valued relational width.
We first state the result on linear equations in an Abelian group and then
discuss the precise meaning of ``simulation''.

Let $\mathcal{G}$ be an Abelian group over a finite set $G$ and let $r \geq 1$ be an integer.
Denote by $E_{\mathcal{G},r}$ the crisp constraint language over domain $G$ with,
for every $a \in G$, and $1 \leq m \leq r$, a relation 
$R^m_a = \{ (x_1, \dots, x_m) \in G^m \mid x_1 + \dots + x_m = a \}$.
In Section~\ref{sec:gap}, we prove the following.

\begin{theorem}\label{thm:eg3notbw}
Let $\mathcal{G}$ be a finite non-trivial Abelian group.
Then, the constraint language $E_{\mathcal{G},3}$ does not have bounded valued relational width.
\end{theorem}


\begin{definition} \label{def:expres}
We say that an $m$-ary weighted relation $\phi$ is \emph{expressible} over a valued constraint
language $\Gamma$ if there exists an instance $I$ of $\VCSP(\Gamma)$ 
with variables $x_1,\ldots,x_m,v_1,\ldots,v_p$
such that
\begin{equation}
\phi(x_1,\ldots,x_m) = \min_{v_1,\dots,v_p} \phi_I(x_1,\dots,x_m,v_1,\dots,v_p).
\end{equation}
\end{definition}

For a fixed set $D$, let $\eq{D}$ denote the binary equality relation $\{ (x,x) \mid x \in D \}$.
Denote by $\langle \Gamma \rangle$ all weighted relations expressible in $\Gamma \cup \{ \eq{D} \}$,
where $D$ is the domain of $\Gamma$.
A weighted relation being expressible over $\Gamma \cup \{ \eq{D} \}$ is the analogue of a
relation being definable by a \emph{primitive positive (pp)} formula (using existential quantification and
conjunction) over a relational structure with equality.
Indeed, when $\Gamma$ is crisp, the two notions coincide.

 
\begin{definition} 
Let $\Delta$ and $\Delta'$ be valued constraint languages on domain $D$ and $D'$, respectively.
We say that
$\Delta$ has an \emph{interpretation} in $\Delta'$ with parameters $(d,S,h)$ if there exists a
$d \in \N$, a set $S\subseteq D'^d$, and a surjective map $h:S\to D$ such
that $\langle \Delta' \rangle$ contains the following weighted relations:
\begin{itemize}
\item $\phi_S \colon D'^d \to \qq$ defined by
$\phi_S(\tup{x}) = 0$ if $\tup{x} \in S$ and $\phi_S(\tup{x}) = \infty$ otherwise;
\item $h^{-1}(\eq{D})$; and
\item $h^{-1}(\phi_i)$, for every weighted relation $\phi_i \in \Delta$, 
\end{itemize}
where $h^{-1}(\phi_i)$, for an $m$-ary weighted relation $\phi_i$,
is the $dm$-ary weighted relation on $D'$ defined by
$h^{-1}(\phi_i)(\tup{x}_1,\dots,\tup{x}_m) = \phi_i(h(\tup{x}_1),\dots,h(\tup{x}_m))$,
for all $\tup{x}_1,\dots,\tup{x}_m\in~S$.
\end{definition}

When $\Gamma$ is crisp, the notion of an interpretation coincides
with the notion of a \emph{pp-interpretation} for relational structures~\cite{Bodirsky08:survey}.

\begin{theorem}\label{thm:clarity}
Let $\Delta$ be a crisp constraint language of finite size
that contains all constant unary relations.
If $\pol(\Delta)$ does not satisfy the BWC, then there exists a finite non-trivial Abelian 
group $\mathcal{G}$
such that $\Delta$ interprets $E_{\mathcal{G},r}$, for every $r \geq 1$.
\end{theorem}

\begin{proof}
It has been shown in~\cite[Theorem~1.6~(4)]{Kozik15:au} that if the polymorphism
algebra $B$ of $\Delta$ does not satisfy the BWC, then the variety generated by
B admits type $\mathbf{1}$ or $\mathbf{2}$ (the notion of admitting types comes
from Tame Congruence Theory~\cite{Hobby88:structure}). By~\cite[Lemmas 20 and
21]{Atseriasetal09:tcs}, this implies that there exists a finite non-trivial
Abelian group $\mathcal{G}$ such that the variety generated by $B$ contains a
reduct $A$ of the polymorphism algebra of $E_{\mathcal{G},r}$, for every $r \geq
1$. For finite algebras $A$ and $B$, $A$ is contained in the variety generated
by $B$ if, and only if, $A$ is contained in the pseudovariety generated by $B$. In
terms of pp-interpretations~\cite{Bodirsky08:survey}, this is equivalent to
$E_{\mathcal{G},r}$ having a pp-interpretation in $\Delta$ (see
also~\cite{Bulatov05:classifying}). 
\end{proof}

Our notion of reduction will be the $\reducesto$ reduction from Definition~\ref{def:reducesto}.

The following theorem shows that
we can augment a valued constraint language with various additional weighted relations.
The transformations in these reductions have previously been used to prove
polynomial-time
reductions~\cite{Bulatov05:classifying,cccjz13:sicomp,tz16:jacm,Kozik15:icalp,tz15:icalp}.
Here, we show that they all \emph{additionally} preserve bounded valued relational width.

\begin{theorem}\label{thm:reductions}
Let $\Gamma$ be a valued constraint language of finite size on domain $D$.
The following holds:
\begin{enumerate}
\item\label{red:express}
If $\phi$ is expressible in $\Gamma$, then $\Gamma \cup \{\phi\} \reducesto \Gamma$.
\item\label{red:equality}
$\Gamma \cup \{\eq{D}\} \reducesto \Gamma$.
\item\label{red:interpret}
If $\Gamma$ interprets the valued constraint language $\Delta$ of finite size, then $\Delta \reducesto \Gamma$.
\item\label{red:feasopt}
If $\phi \in \Gamma$, then 
$\Gamma \cup \{\opt(\phi)\} \reducesto \Gamma$ and
$\Gamma \cup \{\feas(\phi)\} \reducesto \Gamma$.
\item\label{red:coreconstants} If $\Gamma'$ is a core of $\Gamma$ on domain $D' \subseteq D$, then $\Gamma' \cup \mathcal{C}_{D'} \reducesto \Gamma$.
\end{enumerate}
\end{theorem}

Note that Theorem~\ref{thm:reductions}(\ref{red:coreconstants}) is just a
restatement of Lemma~\ref{lem:coreconstants}.

A formal proof of Theorem~\ref{thm:reductions} is given in Section~\ref{sec:reductions}. Here is the main idea.

All of the reductions are based on replacing each constraint $\phi_i(\tup{x}_i)$ of an
instance $I$ of the left-hand side by some gadget, given as an instance $J_i$ of the
right-hand side. The instance $J$ is then defined as the sum of all objective functions $\phi_{J_i}$.

If the replacements satisfy certain conditions, then we show that, for any $1
\leq k' \leq \ell'$, there exist $1 \leq k \leq \ell$ such that if $\Delta'$ has
valued relational width $(k',\ell')$, then $\Delta$ has valued relational width
$(k,\ell)$, so the reductions preserve bounded valued relational width. The
conditions are: $(a)$ for every satisfying and optimal solution $\alpha$ of $J$, there is a
satisfying assignment $\sigma^\alpha$ of $I$ such that $\vcspval(I,
\sigma^\alpha) \leq \vcspval(J,\alpha)$; $(b)$ for every large enough $k$,
feasible solution $\lambda$ to the SA$(k,2k)$-relaxation of $I$, and assignment
$\sigma \colon X_i \to D$ with positive support in $\lambda$, there exists a
satisfying assignment $\alpha_i^\sigma$ of $J_i$ such that
$\phi_i(\sigma(\tup{x}_i)) \geq \vcspval(J_i, \alpha^\sigma_i)$; and $(c)$ the
assignments $\alpha_i^\sigma$ are ``pairwise consistent'', i.e.\
$\alpha^{\sigma_i}_i$ and $\alpha^{\sigma_r}_r$ agree on the intersection of the
variables of $J_i$ and $J_r$ whenever $\sigma_i$ and $\sigma_r$ are restrictions
of some $\sigma \colon X \to D$ with positive support in $\lambda$.

We will also need the following technical lemmas.

\begin{lemma}\label{lem:killingf}
Let $\Gamma$ be a valued constraint language of finite size over domain $D$ and let $F$ be a
finite set of operations over $D$.
If $\supp(\Gamma) \cap F = \emptyset$, then there exists a crisp constraint language $\Delta$
such that
$\pol(\Delta) \cap F = \emptyset$ and $\Delta \reducesto \Gamma$.
\end{lemma}

\begin{proof}
By Lemma~\ref{lem:killing}, for each $f \in F \cap \pol(\Gamma)$,
there is an instance $I_f$ of $\VCSP(\Gamma)$
such that $f \not\in \pol(\opt(\phi_{I_f}))$.
Let $\Delta = \{ \opt(\phi_{I_f}) \mid f \in F \} \cup \{ \feas(\phi) \mid \phi \in \Gamma \}$.
For $f \in F \cap \pol(\Gamma)$, we have $f \not\in \pol(\opt(\phi_{I_f})) \supseteq \pol(\Delta)$.
For $f \in F \setminus \pol(\Gamma)$, we have $f \not\in \pol(\phi)$, for some $\phi \in \Gamma$,
so $f \not\in \pol(\Delta)$.
It follows that $\pol(\Delta) \cap F = \emptyset$.
Finally, $\Delta \reducesto \Gamma$ holds by
Theorem~\ref{thm:reductions}(\ref{red:express}) and (\ref{red:feasopt}).
\end{proof}

\begin{lemma}\label{lem:crispwnus}
Let $\Gamma$ be a valued constraint language of finite size.
If $\supp(\Gamma)$ does not satisfy the BWC,
then there is a crisp constraint language $\Delta$ of finite size 
such that $\pol(\Delta)$ does not satisfy the BWC, and $\Delta \reducesto \Gamma$.
\end{lemma}

\begin{proof}
Since $\supp(\Gamma)$ does not satisfy the BWC, there exists an $m \geq 3$ such that
$\supp(\Gamma)$ does not contain any $m$-ary WNU.
Let $F$ be the (finite) set of all $m$-ary WNUs.
The result follows by applying Lemma~\ref{lem:killingf}
to $\Gamma$ and $F$.
\end{proof}

We now have everything that is needed to prove Theorem~\ref{thm:mainicalp2}.

\begin{proof}[Proof of Theorem~\ref{thm:mainicalp2}]
Suppose that $\supp(\Gamma)$ does not satisfy the BWC.
By Lemma~\ref{lem:crispwnus}, there exists a crisp constraint language $\Delta$ 
such that $\pol(\Delta)$ does not satisfy the BWC and $\Delta \reducesto \Gamma$.
Since $\mathcal{C}_{D} \subseteq \Gamma$, we may assume, without loss of generality, that
$\mathcal{C}_{D} \subseteq \Delta$.

By Theorem~\ref{thm:clarity}, there exists a finite non-trivial Abelian group $\mathcal{G}$
and an interpretation of $E_{\mathcal{G},3}$ in $\Delta$.
By Theorem~\ref{thm:eg3notbw}, $E_{\mathcal{G},3}$ does not have bounded valued relational width.
By Theorem~\ref{thm:reductions}(\ref{red:interpret}),
we have $E_{\mathcal{G},3} \reducesto \Delta \reducesto \Gamma$,
so $\Gamma$ does not have bounded valued relational width.
\end{proof}

\section{Reductions: Proof of Theorem~\ref{thm:reductions}}
\label{sec:reductions}

Theorem~\ref{thm:reductions} follows from
Lemmas~\ref{lem:expr-preserve}--\ref{lem:core-preserve} proved in this section.

For a valued constraint language $\Gamma$, let $\ar(\Gamma)$ denote $\max \{
\ar(\phi) \mid \phi \in \Gamma \}$.

It will sometimes be convenient to add null constraints to
a VCSP instance as placeholders, to ensure that they have a scope,
even if these relations may not necessarily be members of
the corresponding constraint language $\Gamma$.
In order to obtain an equivalent instance that is formally in
$\VCSP\Gamma$), the null constraints can simply be dropped,
as they are always satisfied and do not influence the value of the
objective function.

We extend the convention of denoting the set of variables in $\tup{x}_i$ by $X_i$
to tuples $\tup{y}_i$, $\tup{y}'_i$, and $\tup{v}$, whose sets are denoted by
$Y_i$, $Y'_i$, and $V_i$, respectively.

The following technical lemma is the basis for most of the reductions.

\begin{lemma}\label{lem:sa-reduc}
Let $\Delta$ and $\Delta'$ be valued constraint languages of finite size over domains $D$ and $D'$,
respectively.

Let $(I, i) \mapsto J_i$ be a map 
that to each instance $I$
of $\VCSP(\Delta)$ with variables $V$ and objective function $\sum_{i=1}^q \phi_i(\tup{x}_i)$,
and index $i \in [q]$,
associates an instance $J_i$ of $\VCSP(\Delta')$ with variables $Y_i$
and objective function $\phi_{J_i}$.
Let $J$ be the $\VCSP(\Delta')$ instance with variables $V' = \bigcup_{i=1}^q Y_i$ and
objective function $\sum_{i=1}^q \phi_{J_i}$.

Suppose that the following holds:
\begin{enumerate}[(a)]
\item\label{cond:jtoi}
For every satisfying and optimal assignment $\alpha$ of $J$, there exists a satisfying assignment $\sigma^{\alpha}$
of $I$ such that
\[
\vcspval(I,\sigma^{\alpha}) \leq \vcspval(J,\alpha).
\]
\end{enumerate}

Furthermore, suppose that for
$k \geq \ar(\Delta)$, and
any feasible solution $\lambda$ of the SA$(k,2k)$-relaxation of $I$,
the following properties hold:
\begin{enumerate}[(a)]
\setcounter{enumi}{1}
\item\label{cond:itoj}
For $i \in [q]$, and $\sigma \colon X_i \to D$ with positive support in $\lambda$,
there exists a satisfying assignment $\alpha^{\sigma}_i$ of $J_i$ such that
\[
\phi_i(\sigma(\tup{x}_i)) \geq \vcspval(J_i, \alpha^{\sigma}_i);
\]
\item\label{cond:consistent}
for $i, r \in [q]$, any $X \subseteq V$ with $X_i \cup X_r \subseteq X$, and $\sigma \colon X \to D$ with positive support in $\lambda$, 
\[
\alpha^{\sigma_i}_i|_{Y_i \cap Y_r} = 
\alpha^{\sigma_r}_r|_{Y_i \cap Y_r},
\]
where
$\sigma_i = \Crestrict{\sigma}{X_i}$ and $\sigma_r = \Crestrict{\sigma}{X_r}$.
\end{enumerate}

Then, $I \mapsto J$ is a many-one reduction from $\VCSP(\Delta)$ to $\VCSP(\Delta')$,
and for any $1 \leq k' \leq \ell'$, there exist $1 \leq k \leq \ell$ such that if
$I$ is a gap instance for SA$(k,\ell)$, then
$J$ is a gap instance for SA$(k',\ell')$.
In particular, the reduction preserves bounded valued relational width.
\end{lemma}

\begin{proof}
First, we show that $\vcspopt(I) = \vcspopt(J)$.
From condition (\ref{cond:jtoi}), if $J$ is satisfiable, then so is $I$ and 
$\vcspopt(I) \leq \vcspopt(J)$.
Conversely, if $I$ is satisfiable, and $\sigma$ is an optimal assignment to $I$, then
the SA$(k,2k)$ 
solution $\lambda$ that assigns probability $1$ to $\Crestrict{\sigma}{X}$ for every
$X \subseteq V$ with $|X| \leq 2k$ is feasible.
Let $\sigma_i = \Crestrict{\sigma}{X_i}$.
By (\ref{cond:itoj}), there exist 
satisfying assignments $\alpha_i^{\sigma_i}$ of $J_i$, for all $i \in [q]$, such that
$\vcspopt(I) \geq \lpopt{I}{k}{2k} \geq \sum_{i \in [q]} \vcspval(J_i,\alpha_i^{\sigma_i})$.
Define an assignment $\alpha \colon V' \to D'$ by letting
$\alpha(y) = \alpha_i^{\sigma_i}(y)$ for an arbitrary $i$ such that $y \in Y_i$.
We claim that $\Crestrict{\alpha}{Y_i} = \alpha_i^{\sigma_i}$, for all $i\in
[q]$. From this it follows that $\alpha$ is a satisfying assignment to $J$ such
that
$\sum_{i \in [q]} \vcspval(J_i,\alpha_i^{\sigma_i}) = \vcspval(J,\alpha) \geq \vcspopt(J)$,
and hence that $\vcspopt(I) \geq \vcspopt(J)$.
Indeed, let $y \in V'$ and assume that $y \in Y_{i}$ and $y \in Y_{r}$.
Let $X = X_i \cup X_r$.
Then, since $\lambda(\Crestrict{\sigma}{X}) = 1$, 
it follows from (\ref{cond:consistent}) that $\alpha_i^{\sigma_i}(y) = \alpha_r^{\sigma_r}(y)$.

Let $1 \leq k' \leq \ell'$ be arbitrary, and let $k = \max \{ \ell', \ar(\Delta') \} \cdot \ar(\Delta)$, $\ell = 2k$.
Assume that $I$ is a gap instance for the SA$(k,\ell)$-relaxation of $\VCSP(\Delta)$,
and let $\lambda$ be a feasible solution
such that
$\lpval{I}{\lambda}{k}{\ell} < \vcspopt(I)$ (where $\vcspopt(I)$ may be $\infty$, i.e.\ $I$ may be unsatisfiable).
We show that there is a feasible solution $\kappa$ to the SA$(k',\ell')$-relaxation of $J$ such that
$\lpval{J}{\kappa}{k'}{\ell'} \leq \lpval{I}{\lambda}{k}{\ell}$.
Then, by condition (\ref{cond:jtoi}),
we have $\lpval{J}{\kappa}{k'}{\ell'} \leq \lpval{I}{\lambda}{k}{\ell} < \vcspopt(I) \leq \vcspopt(J)$,
so $J$ is a gap instance for the SA$(k',\ell')$-relaxation of $\VCSP(\Delta')$.
Since $k'$ and $\ell'$ were chosen arbitrarily, the result then follows.

To this end,
augment $I$ with null constraints on $X_{q+1}, \dots, X_{q'}$ so that
for every at most $\ell$-subset $X \subseteq V$, there exists an $i \in [q']$
such that $X_i = X$.
Rewrite the objective function of $J$ as $\sum_{j=1}^{p} \phi'_j(\tup{y}'_j)$,
$\phi' \in \Delta'$,
where,
by possibly first adding extra null constraints to $J$,
we will assume that for every at most $\ell'$-subset $Y \subseteq V'$, there exists a $j \in [p]$
such that $Y'_j = Y$.
For each $i \in [q]$, let $C_i$ be the set of indices $j \in [p]$ corresponding to the weighted constraints in the instance $J_i$.

For $X \subseteq V$, define $Y_X = \bigcup_{i \in [q] : X_i \subseteq X} Y_i$.
For $i \in [q'] \setminus [q]$, let $J_i$ be an instance on the variables $Y_{X_i}$ containing
a single null constraint on the variables.
For $\sigma \in \supp(\lambda_i)$,
and any $r, s \in [q]$ such that $X_r \cup X_s \subseteq X_i$ and $y \in Y_r \cap Y_s$,
by (\ref{cond:consistent}),
it holds that $\alpha^{\sigma_r}_r(y) = \alpha^{\sigma_s}_s(y)$.
Therefore, we can uniquely define $\alpha^{\sigma}_i \colon Y_{X_i} \to D'$ 
for $i \in [q'] \setminus [q]$ by letting
$\alpha^{\sigma}_i(y) = \alpha^{\sigma_r}_r(y)$ for any choice of $r \in [q]$ with
$X_r \subseteq X_i$ and $y \in Y_r$.
Furthermore, this definition is consistent with $\alpha^\sigma_i$ for $i \in [q]$ in the sense
that $(\ref{cond:consistent})$ now holds for all $i, r \in [q']$.

For $m \geq 1$, let 
$X_{(\leq m)} = \{ X = \bigcup_{i \in S} X_i \mid S \subseteq [q], \left| X \right| \leq m \}$, 
and for $Y \subseteq V'$ with $\left|Y\right| \leq \ell'$, let
$X_{(\leq m)}(Y) = \{ X \in X_{(\leq m)} \mid Y \subseteq Y_X \}$.

Let $j \in [p]$ be arbitrary and let $X = \bigcup_{i \in S} X_i \in X_{(\leq n)}(Y'_j)$,
for some $S \subseteq [q]$, where $n=|V|$.
For each $y \in Y'_j$, let $i(y) \in S$ be an index such that $y \in Y_{i(y)}$
and let $X' = \bigcup_{y \in Y'_j} X_{i(y)}$.
Then, $Y'_j \subseteq Y_{X'}$, $X' \subseteq X$,
and $\left|X'\right| \leq \max\{\ell',\ar(\Delta')\} \cdot \ar(\Delta) \leq k$,
so $X' \in X_{(\leq k)}(Y'_j)$.

In other words,
\begin{equation}\label{eq:thindown}
\text{for all } X \in X_{(\leq n)}(Y'_j), \text{there exists } i\in[q'] \text{ such that } X_i \subseteq X \text{ and } X_i \in X_{(\leq k)}(Y'_j).
\end{equation}
In particular (\ref{eq:thindown}) shows that for every $j$ there exists
$i\in[q']$ such that $X_i \in X_{(\leq \ell)}(Y'_j)$,
since $\bigcup_{i \in [q]} X_i \in X_{(\leq n)}(Y'_j)$ for all $j$.

For $j \in [p]$, $\alpha \colon Y'_j \to D'$, and an $i \in [q']$ such that $X_i \in X_{(\leq \ell)}(Y'_j)$, define
\begin{equation}\label{eq:kappadef}
\mu^i_j(\alpha) = \sum_{\substack{\sigma \in \supp(\lambda_i)\\ 
\alpha^{\sigma}_{i}|_{Y'_j} = \alpha}} \lambda_i(\sigma).
\end{equation}

\medskip
\noindent
{\bf Claim:}
Definition (\ref{eq:kappadef}) is independent of the choice of $X_i \in X_{(\leq \ell)}(Y'_j)$.
\medskip

First, we prove this equality for $X_r \subseteq X_i$ with $X_r \in X_{(\leq k)}(Y'_j)$ and $X_i \in X_{(\leq \ell)}(Y'_j)$.

We have,
\[
\mu^{r}_j(\alpha)
=
\sum_{\substack{\tau \in \supp(\lambda_r)\\ \alpha^{\tau}_{r}|_{Y'_j} = \alpha}} \enspace
\sum_{\substack{\sigma \in \supp(\lambda_i)\\ \Crestrict{\sigma}{X_r} = \tau}} \lambda_{i}(\sigma)
=
\sum_{\substack{\sigma \in \supp(\lambda_i) \\ \alpha^{\sigma_r}_r|_{Y'_j} = \alpha}} \lambda_i(\sigma)
=
\sum_{\substack{\sigma \in \supp(\lambda_i)\\ \alpha^{\sigma}_{i}|_{Y'_j} = \alpha}} \lambda_{i}(\sigma)
=
\mu^{i}_j(\alpha),
\]
where the first equality follows by (\ref{eq:kappadef}) and (\ref{sa:marginal}) for $\lambda$ since $\left|X_r\right| \leq k$,
and the second equality follows by interchanging the order of summation and noting that
$\sigma \in \supp(\lambda_i)$ implies that $\sigma_r = \Crestrict{\sigma}{X_r} \in \supp(\lambda_r)$, again
by (\ref{sa:marginal}) for $\lambda$.
The third equality follows by (\ref{cond:consistent}) extended to $i, r \in [q']$.

Next, let $X_r \in X_{(\leq k)}(Y'_j)$ and $X_i \in X_{(\leq \ell)}(Y'_j)$ be arbitrary.
From (\ref{eq:thindown}), it follows that $X_i$ contains a subset $X_s \in X_{(\leq k)}(Y'_j)$.
Since $\left| X_r \cup X_s \right| \leq 2k = \ell$,
there exists an index $u$ such that $X_{u} = X_r \cup X_s$.
The claim now follows by a repeated application of
the first case: $\mu^{r}_j = \mu^{u}_j = \mu^{s}_j = \mu^{i}_j$.

By the claim, we can pick an arbitrary $X_i \in X_{(\leq \ell)}(Y'_j)$ and
uniquely define $\kappa_j = \mu^i_j$.
We now show that this definition of $\kappa$ 
satisfies the equations (\ref{sa:marginal})--(\ref{sa:nonnegative}).

\begin{itemize}

\item
To verify that the equations (\ref{sa:marginal}) hold, let $s, j \in [p]$ be such that $Y'_s \subseteq Y'_j$, and $\beta \colon Y'_s \to D'$.
Let $X_i \in X_{(\leq k)}(Y'_j)$.
We now have:
\[
\sum_{\substack{\alpha \colon Y'_j \to D'\\ \Crestrict{\alpha}{Y'_s} = \beta}} \kappa_j(\alpha)
=\sum_{\substack{\alpha \colon Y'_j \to D'\\ \Crestrict{\alpha}{Y'_s} = \beta}} \mu_j^i(\alpha)
=\mu_s^i(\beta) = \kappa_s(\beta),
\]
where the second equality follows from $Y'_s \subseteq Y'_j$ and a rearrangement of terms,
and the last equality follows from the claim since $Y'_s \subseteq Y'_j \subseteq  Y_{X_i}$,
so $X_i \in X_{(\leq k)}(Y'_s)$.

\item
To verify that the equations (\ref{sa:sum1}) hold
let $Y'_j=\{y\}$ be a singleton and let $X_i\in X_{(\leq
\ell)}(Y'_j)$. We have:
\[
\sum_{\alpha \colon Y'_j \to D'} \kappa_j(\alpha)
=
\sum_{\alpha \colon Y'_j \to D'} 
\sum_{\substack{\sigma \in \supp(\lambda_i)\\ \alpha^{\sigma}_{i}|_{Y'_j} = \alpha}} \lambda_i(\sigma)
=
\sum_{\sigma \in \supp(\lambda_i)} \lambda_i(\sigma)
=
1,
\]
where the last equality follows from~(\ref{sa:sum1}) for $\lambda_i$.

\item
The equations (\ref{sa:infeas}) hold trivially if $\phi'_j$ is a null constraint.
Otherwise, $j \in C_i$ for some $i \in [q]$.
This implies that $X_i \in X_{(\leq k)}(Y'_j)$, and by the claim that $\kappa_j = \mu_j^i$.
Then, $\alpha \in \supp(\kappa_j)$ implies that there is a $\sigma \in \supp(\lambda_i)$
such that $\alpha^{\sigma}_{i}|_{Y'_j} = \alpha$.
By condition (\ref{cond:itoj}) and equation (\ref{sa:infeas}) for $\lambda_i$,
the tuple $\alpha^{\sigma}_{i}(\tup{y}'_j) \in \feas(\phi'_j)$,
so $\kappa_j$ satisfies (\ref{sa:infeas}).

\item $\kappa_j=\mu^i_j$ is defined as a sum of $\lambda$'s, which are nonnegative
by~(\ref{sa:nonnegative}), and thus also satisfies~(\ref{sa:nonnegative}).
\end{itemize}
We conclude that $\kappa$ is a feasible solution to the SA$(k',l')$-relaxation of $J$.

Let $i \in [q]$ and note that by the claim, for every $j \in C_i$, we have $\kappa_j = \mu_j^i$.
Therefore,
\begin{equation}
\begin{aligned}
\sum_{j \in C_i} \sum_{\alpha \in \feas(\phi'_{j})} \kappa_{j}(\alpha) \phi'_j(\alpha(\tup{y}'_j)) &=
\sum_{j \in C_i} \sum_{\alpha \in \feas(\phi'_{j})} \sum_{\substack{\sigma \in \supp(\lambda_i)\\ \alpha^{\sigma}_{i}|_{Y'_j} = \alpha}} \lambda_{i}(\sigma) \phi'_j(\alpha(\tup{y}'_j))\\
&=\sum_{\sigma \in \supp(\lambda_i)} \lambda_{i}(\sigma)
\sum_{j \in C_i}
\sum_{\substack{\alpha \in \feas(\phi'_j)\\ \alpha^{\sigma}_{i}|_{Y'_j} = \alpha}} \phi'_j(\alpha(\tup{y}'_j))\\
&=
\sum_{\sigma \in \supp(\lambda_i)} \lambda_{i}(\sigma)
\sum_{j \in C_i}
\phi'_j(\alpha^{\sigma}_{i}(\tup{y}'_j)) \leq
\sum_{\sigma \in \supp(\lambda_i)} \lambda_{i}(\sigma) \phi_i(\sigma),
\label{eq:objfn}
\end{aligned}
\end{equation}
where the inequality follows from assumption (\ref{cond:itoj}).
Summing inequality (\ref{eq:objfn}) over $i \in [q]$ shows that 
$\lpval{(J}{\kappa}{k'}{\ell'} \leq \lpval{I}{\lambda}{k}{\ell}$ and the lemma follows.
\end{proof}

\begin{lemma}\label{lem:expr-preserve}
Let $\Gamma$ be a valued constraint language of finite size
and let $\phi$ be a weighted relation
expressible over $\Gamma$.
Then, $\Gamma \cup \{ \phi \} \reducesto \Gamma$.
\end{lemma}

\begin{proof}
Let $I$ be an instance of $\VCSP(\Gamma \cup \{ \phi \})$
with variables $V = \{x_1,\dots,x_n\}$ and objective function
$\phi_{I}(x_1, \dots, x_n) = \sum_{i=1}^{q} \phi_i(\tup{x}_i)$,
where $\phi_i \in \Gamma \cup \{ \phi \}$ and $\tup{x}_i$ is such that $X_i \subseteq V$.
Let $I'$ be an instance of $\VCSP(\Gamma)$ such that 
$\phi(x_1, \dots, x_m) = \min_{v_i \in D}\phi_{I'}(x_1, \dots, x_m, v_1, \dots, v_p)$.

For $i \in [q]$ such that $\phi_i \in \Gamma$, let
$J_i$ be the instance on variables $Y_i = X_i$ with $\phi_{J_i}(Y_i) = \phi_i(\tup{x}_i)$.
For $i \in [q]$ such that $\phi_i = \phi$, let
$\tup{v}_i$ be a copy of the variables $v_1, \dots, v_p$,
and let
$J_i$ be the instance on variables $Y_i = X_i \cup V_i$ with
objective function $\phi_{J_i}(Y_i) = \phi_{I'}(\tup{x}_i, \tup{v}_i)$.
Let $J$ be the $\VCSP(\Gamma)$ instance with variables $\bigcup_i Y_i$
and objective function $\sum_i \phi_{J_i}$.

We verify properties $(a)$--$(c)$ of Lemma~\ref{lem:sa-reduc}.

\medskip
($a$)
Let $\alpha$ be any satisfying assignment of $J$ and define $\sigma^\alpha = \Crestrict{\alpha}{V}$.
For $i \in [q]$ such that $\phi_i \in \Gamma$, we have $\phi_i(\sigma^{\alpha}(\tup{x}_i)) = \phi_{J_i}(\alpha(Y_i)) < \infty$.
For $i \in [q]$ such that $\phi_i = \phi$, we have $\phi_i(\sigma^{\alpha}(\tup{x}_i)) \leq 
\phi_{J_i}(\alpha(Y_i)) = \phi_{I'}(\alpha(\tup{x}_i, \tup{v}_i)) < \infty$.
Summing over all $i \in [q]$ gives $\vcspval(I,\sigma^\alpha) \leq \vcspval(J,\alpha) < \infty$.

\medskip
Let $k \geq \ar(\Gamma \cup \{\phi\})$ and suppose that $\lambda$ is a feasible solution to the SA$(k,2k)$-relaxation of $I$.
For all $i \in [q]$ and $\sigma \colon X_i \to D$ with positive support in $\lambda$,
define $\alpha^{\sigma}_i$ as follows.
If $\phi_i \in \Gamma$, then define $\alpha^{\sigma}_i = \sigma$.
Otherwise, $\phi_i = \phi$.
Let
$\gamma^{\sigma}_{i} \colon V_i \to D$ be any assignment such that
$\phi_i(\sigma(\tup{x}_i)) =
\phi_J(\sigma(\tup{x}_i), \gamma^{\sigma}_{i}(\tup{v}_j))$,
and define
$\alpha^{\sigma}_i \colon X_i \cup V_i \to D$ by
$\alpha^{\sigma}_i = \sigma \cup \gamma^{\sigma}_i$.

\medskip
($b$)
For all $i \in [q]$,
$\vcspval(J_i, \alpha^{\sigma}_i) =
\phi_{J_i}(\alpha^{\sigma}_i(Y_i)) =
\phi_i(\sigma(\tup{x}_i)) < \infty$, 
where the equalities hold by construction,
and the inequality follows from the feasibility of $\lambda$.

\medskip
($c$)
Let $i, r \in [q]$ and $X \subseteq V$ be as in the lemma and
suppose that $\sigma \colon X \to D$ has positive support in $\lambda$.
If $i = r$, then there is nothing to show.
Otherwise, $Y_i \cap Y_r = X_i \cap X_r$, so
$\alpha^{\sigma_i}_i|_{Y_i \cap Y_r} = 
\sigma_i|_{X_i \cap X_r} = \sigma|_{X_i \cap X_r} = \sigma_r|_{X_i \cap X_r} =
\alpha^{\sigma_r}_r|_{Y_i \cap Y_r}$.

It follows that Lemma~\ref{lem:sa-reduc} is applicable, so
$\Gamma \cup \{\phi\} \reducesto \Gamma$.
\end{proof}

\begin{lemma}\label{lem:eq-preserve}
Let $\Gamma$ be a valued constraint language of finite size over domain $D$.
Then, $\Gamma \cup \{ \eq{D} \} \reducesto \Gamma$.
\end{lemma}

\begin{proof}
Let $I$ be an instance of $\VCSP(\Gamma \cup \{ \eq{D} \})$ with variables $V = \{x_1, \dots, x_n\}$ and
objective function $\phi_{I}(x_1, \dots, x_n) = \sum_{i=1}^{q} \phi_i(\tup{x}_i)$,
where $\phi_i \in \Gamma \cup \{ \eq{D} \}$ and $\tup{x}_i$ is such that $X_i \subseteq V$.
Define the undirected graph $G = (V,E)$, where
$E$ contains an edge between $u$ and $v$
if, and only if, there is a constraint $\eq{D}(u,v)$ in $I$.
Let $\sim$ be the equivalence relation on
$V$ defined by $u \sim v$ if $u$ and $v$ are in the same connected component of $G$.
For $v \in V$, let $\tilde{v}$ denote the equivalence class of $\sim$
containing $v$.
For a tuple of variables $\tup{x} = (v_1, \dots, v_m) \in V^m$, define $\tilde{\tup{x}} = (\tilde{v}_1, \dots, \tilde{v}_m)$.

For $i \in [q]$, 
let $\tup{y}_i = \tilde{\tup{x}}_i$ and let
$J_i$ be an instance on variables $Y_i$.
If $\phi_i \in \Gamma$‚ then let the objective function be $\phi_{J_i}(Y_i) = \phi_i(\tup{y}_i)$.
Otherwise, let $\phi_{J_i}$ be a null-constraint on $Y_i$.
Let $J$ be the $\VCSP(\Gamma)$ instance with variables $\bigcup_i Y_i$ and objective
function $\sum_i \phi_{J_i}$.

We verify properties $(a)$--$(c)$ of Lemma~\ref{lem:sa-reduc}.

\medskip
($a$)
Let $\alpha$ be satisfying assignment of $J$ and define $\sigma^\alpha(v) = \alpha(\tilde{v})$
for all $v \in V$.
It is clear that 
$\phi_i(\sigma^{\alpha}(\tup{x}_i)) = \phi_{J_i}(\alpha(Y_i))$ for all $i \in [q]$.
Summing over all $i$ gives $\vcspval(I,\sigma^\alpha) = \vcspval(J,\alpha)$.

\medskip
Let $k \geq \ar(\Gamma \cup \{\eq{D}\}) \geq 2$ and suppose that $\lambda$ is a feasible solution to the SA$(k,2k)$-relaxation of $I$.
We claim that,
for all $i \in [q]$ and $\sigma \colon X_i \to D$ with positive support in $\lambda$,
\begin{equation}\label{eq:sigma}
u \sim v \implies \sigma(u) = \sigma(v).
\end{equation}

For $\tilde{v} \in Y_i$, let $\alpha^{\sigma}_i(\tilde{v}) = \sigma(u)$ for some $u \in \tilde{v} \cap X_i$.
By the claim, the definition of $\alpha^{\sigma}_i$ is actually independent of the choice of $u \in \tilde{v} \cap X_i$.
The justification of the claim follows at the end of the proof.

\medskip
($b$)
For all $i \in [q]$,
$\vcspval(J_i, \alpha^{\sigma}_i) =
\phi_{J_i}(\alpha^{\sigma}_i(Y_i)) =
\phi_i(\sigma(\tup{x}_i)) < \infty$, 
where the second equality holds by (\ref{eq:sigma}),
and the inequality follows from the feasibility of $\lambda$.

\medskip
($c$)
Let $i, r \in [q]$ and $X \subseteq V$ be as in the lemma and
suppose that $\sigma \colon X \to D$ has positive support in $\lambda$.
Let $\tilde{v} \in Y_i \cap Y_r$, let $v_1 \in \tilde{v} \cap X_i$ and $v_2 \in \tilde{v} \cap X_r$.
By (\ref{eq:sigma}), $\sigma(v_1) = \sigma(v_2)$, so
$\alpha^{\sigma_i}_i(\tilde{v}) = \sigma_i(v_1) = \sigma(v_1) = \sigma(v_2) = \sigma_r(v_2) = \alpha^{\sigma_r}_r(\tilde{v})$.

\medskip
It remains to prove that (\ref{eq:sigma}) holds for all
$\sigma \colon X_i \to D$ with positive support in $\lambda$.
The proof is by induction over the length of a shortest path,
$u = u_0, \dots, u_d = v$, between $u$ and $v$ in the graph $G$.
If $d = 0$, then $u = v$, so there is nothing to prove.
Assume therefore that $d > 0$ and
that the claim holds for all assignments with positive support and all
$u' \sim v'$ with a shortest path of length strictly smaller
than $d$.
Let $X' = \{u_0, u_d\}$ and note that
since $\left| X' \right| = 2 \leq k$, there exists an assignment $\tau' \colon X' \to D$
with positive support in $\lambda$ such that $\tau' = \Crestrict{\sigma}{X'}$.
Now, let $X = X' \cup \{ u_{d-1} \}$.
Since $\left| X \right| \leq 3 \leq 2k$, it follows that
$\lambda$ has a distribution over assignments to $X$, so there exists an assignment
$\tau \colon X \to D$ with positive support in $\lambda$ such that
$\Crestrict{\tau}{X'} = \tau' = \Crestrict{\sigma}{X'}$.
In particular, $\tau(u_0) = \sigma(u_0)$ and
$\tau(u_d) = \sigma(u_d)$.

By assumption,
there is an equality constraint on $X'' = \{u_{d-1}, u_d\}$ in $J$, so
any assignment $\tau'' \colon X'' \to D$ with positive support in $\lambda$
must have $\tau''(u_{d-1}) = \tau''(u_d)$.
Since equation (\ref{sa:marginal}) holds for $X'' \subseteq X$,
it follows that $\Crestrict{\tau}{X''}$ has positive support in $\lambda$,
and hence $\tau(u_{d-1}) = \tau(u_d)$.

By the induction hypothesis applied to $\tau$ and the path $u_0, \dots, u_{d-1}$,
we now have
$\sigma(u_0) = \tau(u_0) = \tau(u_{d-1}) = \tau(u_d) = \sigma(u_d)$.
It follows that Lemma~\ref{lem:sa-reduc} is applicable, so
$\Gamma \cup \{\phi\} \reducesto \Gamma$.
\end{proof}


\begin{lemma}\label{lem:i-preserve}
Let $\Delta'$ and $\Delta$ be constraint languages of finite size
and assume that $\Delta'$ interprets $\Delta$.
Then, $\Delta \reducesto \Delta'$.
\end{lemma}

\begin{proof}
Let $D$ and $D'$ be the domains of $\Delta$ and $\Delta'$, respectively.
Let $(d, S, h)$ be an interpretation of $\Delta$ in $\Delta'$.
By Lemma~\ref{lem:expr-preserve} and Lemma~\ref{lem:eq-preserve},
we may assume that $\Delta'$ contains
the $d$-ary weighted relation $\phi_S$,
and for each $\phi_i \in \Delta$, that $\Delta'$ contains $h^{-1}(\phi_i)$.

Let $I$ be an instance of $\VCSP(\Delta)$ with variables $V = \{x_1, \dots,
x_n\}$ and
objective function $\phi_I(x_1, \dots, x_n) = \sum_{i=1}^{q} \phi_i(\tup{x}_i)$.
Assume that $\phi_I$ contains a distinguished unary null-constraint for each singleton 
subset $\{x_j\} \subseteq V$, i.e.\ that
for each $1 \leq j \leq n$, there exists an $i \in [q]$ such that $\phi_i$ is a null constraint,
and $\tup{x}_i = (x_j)$.
Let $\tup{v}_1, \dots, \tup{v}_n$ be $d$-tuples of distinct fresh variables
($nd$ distinct variables overall)
and let $V_j$ be the set of variables in $\tup{v}_j$, for each $1 \leq j \leq n$.

For $i \in [q]$,
let $J_i$ be an instance on variables 
$Y_i = \bigcup_{j \colon x_j \in X_i} V_j$.
If $\phi_i(x_j)$ is one of the distinguished null-constraints, 
then let $\phi_{J_i}(Y_i) = \phi_{S}(\tup{v}_j)$.
Otherwise,
assuming $\tup{x}_i=(x_{i_1},\ldots,x_{i_r})$,
let $\phi_{J_i}(Y_i) = h^{-1}(\phi_i)(\tup{v}_{i_1},\ldots,\tup{v}_{i_r})$.
Let $J$ be the $\VCSP(\Delta')$ instance with variables $\bigcup_i Y_i$ and objective function $\sum_i \phi_{J_i}$.

We verify properties $(a)$--$(c)$ of Lemma~\ref{lem:sa-reduc}.

\medskip
($a$)
Let $\alpha$ be any satisfying assignment of $J$ and define $\sigma^\alpha \colon V \to D$
by $\sigma^\alpha(x_j) = h(\alpha(\tup{v}_j))$.
This is well-defined since 
there always is a constraint $\phi_S(\tup{v}_j)$ in $J$ which ensures that
$\alpha(\tup{v}_j)\in S$.
For all $i \in [q]$,
we have
$\phi_i(\sigma^\alpha(\tup{x}_i))=\phi_i(h(\alpha(\tup{v}_{i_1})), \dots, h(\alpha(\tup{v}_{i_r})))=\phi_{J_i}(\alpha(Y_i))$,
where $\tup{x}_i=(x_{i_1},\ldots,x_{i_r})$.
Summing over all $i$ gives $\vcspval(I,\sigma^\alpha) = \vcspval(J,\alpha)$.

\medskip
For each $x_j \in V$ and $a \in D$, let $\tau_{j, a}\colon V_j\to D'$ be an assignment such that
$\tau_{j,a}(\tup{v}_j) \in S$ and $h(\tau_{j,a}(\tup{v}_j))=a$.
Let $k \geq \ar(\Delta)$ and suppose that $\lambda$ is a feasible solution to the SA$(k,2k)$-relaxation of $I$.
For all $i \in [q]$ and $\sigma \colon X_i \to D$ with positive support in $\lambda$,
define $\alpha^{\sigma}_i \colon Y_i \to D'$ by $\alpha^{\sigma}_i(v) = \tau_{j,\sigma(x_j)}(v)$ where $j$ is the index such
that $v \in V_j$, i.e., $\alpha^{\sigma}_i = \bigcup_{j \colon x_j \in X_i} \tau_{j,\sigma(x_j)}$.

\medskip
($b$)
For all $i \in [q]$, assuming $\tup{x}_i=(x_{i_1},\ldots,x_{i_r})$, 
we have
$\vcspval(J_i, \alpha^{\sigma}_i) =
\phi_{J_i}(\alpha^{\sigma}_i(Y_i)) =
h^{-1}(\phi_i)(\alpha^{\sigma}_i(\tup{v}_{i_1}), \dots, \alpha^{\sigma}_i(\tup{v}_{i_r})) =
\phi_i(h(\alpha^{\sigma}_i(\tup{v}_{i_1})), \dots, h(\alpha^{\sigma}_i(\tup{v}_{i_r}))) =
\phi_i(\sigma(\tup{x}_i)) < \infty$, 
where the inequality follows from the feasibility of $\lambda$.

\medskip
($c$)
Let $i, r \in [q]$ and $X \subseteq V$ be as in Lemma~\ref{lem:sa-reduc} and suppose
that $\sigma \colon X \to D$ has positive support in $\lambda$.
Let $v \in Y_i \cap Y_r$ and let $\tup{v}_j$ be the tuple of variables that contains $v$.
Then, $x_j \in X_i \cap X_r$, so
$\alpha^{\sigma_i}_i(v) = \tau_{j, \sigma(x_j)}(v) = \alpha^{\sigma_r}_r(v)$.

It follows that Lemma~\ref{lem:sa-reduc} is applicable, so
$\Delta \reducesto \Delta'$.
\end{proof}


\begin{lemma}\label{lem:opt-preserve}
Let $\Gamma$ be a valued constraint language of finite size and $\phi \in \Gamma$.
Then, 
$\Gamma \cup \{ \opt(\phi) \} \reducesto \Gamma$.
\end{lemma}

\begin{proof} To avoid trivial cases, we will assume that all weighted relations
in $\Gamma$ take at least one finite value. Moreover, in order to simplify the
proof, we will assume that $\min(\phi')=0$ for every $\phi'\in\Gamma$. 
This is without loss of generality as
replacing $\phi'$ by $\phi'+c$, for any $c\in\mathbb{Q}$, changes the value of the objective
function of a VCSP instance 
by the same additive constant as the
objective function of the LP relaxation, for all feasible solutions to the
corresponding problems.

Let $I$ be an arbitrary instance of $\VCSP(\Gamma \cup \{ \opt(\phi) \})$ with
variables $V$ and objective function $\sum_{i=1}^{q} \phi_i(\tup{x}_i)$.
We create an instance $J$ of $\VCSP(\Gamma)$ as follows.
The variables of $J$ are the same as the variables in $I$.
Every weighted constraint $\phi_i(\tup{x}_i)$ in $I$, where $\phi_i \neq
\opt(\phi)$ appears also in $J$.
Every weighted constraint $\opt(\phi)(\tup{x}_i)$ is replaced by $C$ copies of
$\phi(\tup{x}_i)$ in $J$,
where the value of the constant $C$ is chosen as follows:
If $\phi$ only takes a single distinct finite value (which we assume is $0$),
then let $C = 1$.
Otherwise,
let $U = \sum_{i=1}^q \max(\phi_i)$, where $\max(\phi_i)$
denotes the largest \emph{finite} value of the weighted relation $\phi_i$.
Let $\delta$ be the smallest non-zero finite value of $\phi$.
Now, let $C = \lceil(U+1)/\delta\rceil$.
$U$ can be computed in polynomial time and
the value of $C$ depends linearly on the number of constraints in $I$, so
the size of $J$ is polynomial in the size of $I$.

\medskip
\noindent
First, we prove that $\vcspopt(J)$ determines $\vcspopt(I)$.
Any satisfying assignment to $I$ is also a satisfying assignment to $J$, so
\begin{equation}\label{eq:jprimesat}
\vcspopt(J) \leq \vcspopt(I).
\end{equation}

If $J$ has a satisfying assignment, then let $\sigma$ be an optimal assignment.
We distinguish two cases.
First, assume that $\sigma$ assigns the optimal zero value to every copy of
$\phi$.
Then, $\sigma$ is also a satisfying assignment of $I$, so
\begin{equation}\label{eq:jopt}
\vcspopt(I) \leq \vcspval(I,\sigma) = \vcspval(J,\sigma) = \vcspopt(J).
\end{equation}
From (\ref{eq:jprimesat}) and (\ref{eq:jopt}), we see that
$\vcspopt(I) \leq \vcspval(I, \sigma) = \vcspopt(J) \leq \vcspopt(I)$, so
$\sigma$ is also an optimal assignment to $I$.

Otherwise, $\sigma$ assigns a sub-optimal value to at least $C$ copies of
$\phi$, so
\begin{equation*}
\vcspval(J,\sigma) \geq C \delta + \vcspopt(I)  \geq  U + 1.
\end{equation*}
In this case, $\vcspopt(J) > U$. But $U \geq \vcspopt(I)$ if $I$ is satisfiable,
which contradicts (\ref{eq:jprimesat}), and hence $I$ is unsatisfiable.
In summary, if $J$ is unsatisfiable, or if $\vcspopt(J) > U$, then $I$ is unsatisfiable,
and otherwise $\vcspopt(I) = \vcspopt(J)$.

\medskip
\noindent
Next, we prove that, for any given parameters $1 \leq k \leq \ell$,
if $\Gamma \cup \{ \opt(\phi) \}$ does not have valued relational width $(k, \ell)$,
then $\Gamma$ does not have valued relational width $(k,\ell)$.
Let $I$ be an instance of $\VCSP(\Gamma \cup \{ \opt{\phi} \})$ and
$\lambda$ a feasible solution to the SA$(k,\ell)$-relaxation of $I$,
with $\lpval{I}{\lambda}{k}{\ell} < \vcspopt(I)$,
where $\vcspopt(I)$ could be $\infty$.
We will assume that $I$ has been augmented with null constraints
so that, for every subset $V' \subseteq V$ with $\left| V' \right| \leq \ell$,
there is some $i \in [q]$ with $X_i = V'$.
Let $J$ be the instance of $\VCSP(\Gamma)$ constructed above.

Let $\lambda'$ be the feasible solution to the SA$(k,\ell)$-relaxation of $J$ obtained from $\lambda$
by letting $\lambda'_j = \lambda_i$ for all $\phi$-constraints of $J$ with index $j$ that were
introduced as copies of the $\opt(\phi)$-constraint of $I$ with index $i$.
Then, $\lambda'$ assigns an optimal value to each $\phi$-constraint, so
$\lpval{J}{\lambda'}{k}{\ell} = \lpval{I}{\lambda}{k}{\ell}$.

If $J$ is unsatisfiable, then $\lpopt{J}{k}{\ell} < \vcspopt(J)$, so $\Gamma$
does not have valued relational width $(k, \ell)$. If $J$ is satisfiable and $I$
is also satisfiable then it was shown above that $\vcspopt(J) = \vcspopt(I)$, so
$\lpopt{J}{k}{\ell} \leq \lpopt{I}{k}{\ell} < \vcspopt(I) = \vcspopt(J)$, and
again $\Gamma$ does not have valued relational width $(k,\ell)$. Finally, if $J$
is satisfiable and $I$ is unsatisfiable then $\vcspopt(J) > U$. Since $\lambda$
is a feasible solution, we have $\lpopt{I}{k}{\ell} \leq U$ from the definition
of $U$. Then, $\lpopt{J}{k}{\ell} \leq \lpopt{I}{k}{\ell} \leq U < \vcspopt(J)$,
so again $\Gamma$ does not have valued relational width $(k,\ell)$.
Since $k$ and $\ell$ were chosen arbitrarily, the result follows. \end{proof}

\begin{lemma}\label{lem:feas-preserve}
Let $\Gamma$ be a valued constraint language of finite size and $\phi \in \Gamma$.
Then, 
$\Gamma \cup \{ \feas(\phi) \} \reducesto \Gamma$.
\end{lemma}

\begin{proof}
To avoid trivial cases, we will assume that all weighted relations in $\Gamma$
take at least one finite value. As in the proof of Lemma~\ref{lem:opt-preserve},
we will assume that $\min(\phi)=0$.
Let $I$ be an arbitrary instance of $\VCSP(\Gamma \cup \{ \feas(\phi) \})$ with
variables $V$ and objective function $\sum_{i=1}^{q} \phi_i(\tup{x}_i)$.
We create an instance $J$ of $\VCSP(\Gamma)$ as follows.
The variables of $J$ are the same as the variables in $I$.
For every weighted constraint $\phi_i(\tup{x}_i)$ in $I$ with $\phi_i \in
\Gamma$,
we add $C$ copies of $\phi_i(\tup{x}_i)$ in $J$.
Every weighted constraint $\feas(\phi)(\tup{x}_i)$ is replaced by
$\phi(\tup{x}_i)$ in $J$.
The value of the constant $C$ is chosen as follows:
If $\phi$ only takes a single distinct finite value, then let $C = 1$.
Otherwise, let $U$ be the largest \emph{finite} value of $\phi$.
Let $\delta = 1/M$ where $M > 0$ is any constant such that $M \cdot \phi_i$ is
integral for every $i$.
This implies that $\delta$ is less than or equal to the least possible
difference between any two satisfying assignments of $I$.
Now, let $C = \lceil N(U+1)/\delta\rceil$,
where $N$ is the number of occurrences of $\feas(\phi)$ in $I$.
The value of $C$ can be computed in polynomial time and
depends linearly on the number of constraints in $I$, so
the size of $J$ is polynomial in the size of $I$.

An assignment $\sigma \colon V \to D$ satisfies $I$ if, and only if, it is
satisfies $J$,
and
\begin{equation}
C \cdot \vcspval(I,\sigma) \leq \vcspval(J,\sigma) \leq C \cdot
\vcspval(I,\sigma)+NU.
\end{equation}

Let $\sigma$ be an optimal assignment to $J$ and suppose that there exists an
assignment $\sigma'$
to $I$ such that $\vcspval(I,\sigma') < \vcspval(I,\sigma)$.
Then,
\begin{align*}
\vcspval(J,\sigma') &\leq C \cdot \vcspval(I,\sigma') + NU \\
&\leq C \cdot (\vcspval(I,\sigma)-\delta) + NU \\
&\leq C \cdot \vcspval(I,\sigma) + NU-C\cdot\delta \\
&< C \cdot \vcspval(I,\sigma) \\
&\leq \vcspval(J,\sigma),
\end{align*}
which contradicts $\sigma$ being optimal.
Hence, $\sigma$ is also an optimal assignment to $I$.

Next, we prove that, for any given parameters $1 \leq k \leq \ell$,
if $\Gamma \cup \{ \feas(\phi) \}$ does not have valued relational width $(k,
\ell)$,
then $\Gamma$ does not have valued relational width $(k,\ell)$.
Let $I$ be an instance of $\VCSP(\Gamma \cup \{ \feas(\phi) \})$ and
$\lambda$ a feasible solution to the SA$(k,\ell)$-relaxation of $I$
with $\lpval{I}{\lambda}{k}{\ell} < \vcspopt(I)$.
We will assume that $I$ has been augmented with null constraints
so that, for every subset $V' \subseteq V$ with $\left| V' \right| \leq \ell$,
there is some $i \in [q]$ with $X_i = V'$.
If $I$ is unsatisfiable, then let $J$ be the instance of $\VCSP(\Gamma)$
constructed as above.
Otherwise, let $\epsilon = \vcspopt(I) - \lpval{I}{\lambda}{k}{\ell} > 0$, and let
$J$ be the instance
constructed as above,
but with $C = \max \{ \lceil N(U+1)/\delta\rceil, \lceil N(U+1)/\epsilon \rceil
\}$.

Let $\lambda'$ be the feasible solution to the SA$(k,\ell)$-relaxation of $J$
obtained from $\lambda$
by letting $\lambda'_j = \lambda_i$ for every constraint of $J$ with index $j$
that was
introduced as a (possibly single) copy of the constraint $\phi_i(\tup{x}_i)$ of
$I$.
Then, $\lpval{J}{\lambda'}{k}{\ell} \leq C \cdot \lpval{I}{\lambda}{k}{\ell} + N U$.

The instance $J$ is unsatisfiable if, and only if, $I$ is unsatisfiable,
and in this case, $\lpopt{J}{k}{\ell} < \vcspopt(J)$, so
$\Gamma$ does not have valued relational width $(k, \ell)$.
Otherwise, $J$ is satisfiable, and $C \cdot \vcspopt(I) \leq \vcspopt(J)$,
so
$\lpopt{J}{k}{\ell}
\leq C \cdot \lpopt{I}{k}{\ell} + NU
\leq C \cdot (\vcspopt(I)-\epsilon) + NU
\leq \vcspopt(J) + NU -C \cdot \epsilon
< \vcspopt(J)$,
so $\Gamma$ does not have valued
relational width $(k,\ell)$.
Since $k$ and $\ell$ were chosen arbitrarily, the result follows.
\end{proof}

\begin{lemma}\label{lem:core-preserve}
Let $\Gamma$ be a valued constraint language of finite size over domain $D$ and
let $\Gamma'$ be a core of $\Gamma$ with $D' \subseteq D$.
Then, $\Gamma' \cup \mathcal{C}_{D'} \reducesto \Gamma$.
\end{lemma}

\begin{proof}
Let $I'$ be an instance of $\VCSP(\Gamma')$, and let $I$ be the instance of $\VCSP(\Gamma)$
obtained from $I'$ by substituting every restricted weighted relation in $\Gamma'$ by
its corresponding weighted relation in $\Gamma$.
Then, by Lemma~\ref{lem:core},
$\vcspopt(I') = \vcspopt(I)$.
Fix $1 \leq k \leq \ell$ and assume that $I'$ is a gap instance for the SA$(k,\ell)$-relaxation of
$\VCSP(\Gamma')$.
Then, $\lpopt{I}{k}{\ell} \leq \lpopt{I'}{k}{\ell} < \vcspopt(I') = \vcspopt(I)$,  where the first inequality follows
from the fact that $I'$ is a restriction of $I$.
This establishes $\Gamma' \reducesto \Gamma$.

Let $F$ be the set of unary operations on $D'$ that are not in $\supp(\Gamma')$
and apply Lemma~\ref{lem:killingf}
to $\Gamma'$ and $F$.
This provides a crisp constraint language $\Delta$ on $D'$
such that $\Delta \reducesto \Gamma'$ and such that
every unary operation in $\pol(\Delta)$ is also in $\supp(\Delta)$. Since
$\Gamma'$ is a core only bijections can occur in $\supp(\Gamma')$. By 
Lemma~\ref{lem:killingf}, $\pol(\Delta)\cap F=\emptyset$ 
and hence only bijections can occur in $\pol(\Delta)$.
Thus $\Delta$ is also a core.
We finish the proof by showing that
$\Gamma' \cup \mathcal{C}_{D'} \reducesto \Gamma' \cup \Delta$, 
using Lemma~\ref{lem:sa-reduc}. Indeed, by Lemma~\ref{lem:killingf} we have
$\Delta \reducesto \Gamma'$ and we have previously shown that
$\Gamma' \reducesto \Gamma$. Overall, $\Gamma' \cup \mathcal{C}_{D'} \reducesto
\Gamma' \cup \Delta \reducesto \Gamma' \reducesto \Gamma$ and thus $\Gamma' \cup
\mathcal{C}_{D'} \reducesto \Gamma$.

Let $I_\Delta$ be the instance on variables $V_\Delta = \{ x_a \mid a \in D\}$ and containing, for every
$\phi \in \Delta$, and $\tup{a} \in D^{\ar(\phi)}$, a constraint $\phi(\tup{x}_\tup{a})$,
where $\tup{x}_\tup{a}[i] = x_{a[i]}$ for $1 \leq i \leq \ar(\phi)$.
Every satisfying assignment $\alpha$ to $I_\Delta$ defines an operation $f_\alpha \colon D \to D$
by the map $a \mapsto \alpha(x_a)$.
The instance $I_\Delta$ is sometimes called the \emph{indicator instance}~\cite{Jeavons97:jacm} and has the following property:
\begin{equation}\label{snyggatill}
\text{$\alpha$ is a satisfying assignment of $I_\Delta$ if, and only if,
$f_\alpha$ is a unary polymorphism of $\Delta$.}
\end{equation}

Let $I$ be an arbitrary instance of $\VCSP(\Gamma' \cup \mathcal{C}_{D'})$ with
variables $V = \{x_1, \dots, x_n\}$ and objective function 
$\phi_I(x_1, \dots, x_n) = \sum_{i=1}^q \phi_i(\tup{x}_i)$, 
where $\phi_i \in \Gamma' \cup \mathcal{C}_{D'}$ and $\tup{x}_i$ such that $X_i \subseteq V$.
Assume without loss of generality that $V \cap V_\Delta = \emptyset$.
For $v \in V$, define $\hat{v} := x_a$ if there is a unary constraint $v = a$ in $I$,
and define $\hat{v} := v$, otherwise.
For a tuple of variables $\tup{x} = (v_1, \dots, v_m) \in V^m$,
define $\hat{\tup{x}} = (\hat{v}_1, \dots, \hat{v}_m)$.

For $i \in [q]$ such that $\phi_i \in \Gamma'$, let $\tup{y}_i = \hat{\tup{x}}_i$ and
let $J_i$ be the instance on variables $Y_i$ with objective function $\phi_{J_i}(Y_i) = \phi_i(\tup{y}_i)$.
For $i \in [q]$ such that $\phi_i(x_i)$ is a unary constraint $x_i = a$,
let $J_i$ be the instance $I_\Delta$ on variables $Y_i = V_\Delta$.
Note that each $J_i$ corresponding to a unary constraint $x_i = a$ is the same instance $I_\Delta$
on the same variables $V_\Delta$.
Let $J$ be the $\VCSP(\Gamma' \cup \Delta)$ instance with variables $\bigcup_i Y_i$
and objective function $\sum_i \phi_{J_i}$.

We verify properties $(a)$--$(c)$ of Lemma~\ref{lem:sa-reduc}.

\medskip
($a$)
Let $\alpha$ be an optimal assignment to $J$ and consider the operation $f_\alpha$
in (\ref{snyggatill}) obtained from the unique copy of $I_\Delta$ in $J$.
Since the unary operations in $\pol(\Delta)$ are bijections and closed under composition,
it follows that $f^{-1}_{\alpha}$ is also in $\pol(\Delta)$ and therefore in $\supp(\Gamma)$.
Hence, by Lemma~\ref{lem:killing}, $\beta := f^{-1}_{\alpha} \circ \alpha$ is also an optimal
assignment to $J$ and $f_\beta$ is the identity operation.
We define $\sigma^{\alpha}(x) = a$ if $\hat{x} = x_a$ for some $a \in D'$,
and $\sigma^{\alpha}(x) = \beta(x)$ otherwise.
All unary constraints $x = a$ in $I$ are satisfied by $\sigma^{\alpha}$ and all other
constraints take the same value as in $J$, hence $\vcspval(I,\sigma^{\alpha}) = \vcspval(J,\alpha)$.

\medskip
Let $k \geq \ar(\Gamma' \cup \mathcal{C}_{D'})$ 
and suppose that $\lambda$ is a feasible solution to the SA$(k,2k)$-relaxation of $I$.
Let $\gamma$ be the satisfying assignment of $I_{\Delta}$ that assigns $a$ to $x_a$ for
all $a \in D'$.
For all $i \in [q]$ and $\sigma \colon X_i \to D$ with positive support in $\lambda$,
define $\alpha^{\sigma}_i = \Crestrict{(\sigma \cup \gamma)}{Y_i}$.

\medskip
($b$)
For all $i \in [q]$,
$\vcspval(J_i, \alpha^{\sigma}_i) =
\phi_{J_i}(\alpha^{\sigma}_i(Y_i)) =
\phi_i(\sigma(\tup{x}_i)) < \infty$, 
where the equalities hold by construction,
and the inequality follows from the feasibility of $\lambda$.

\medskip
($c$)
Let $i, r \in [q]$ and $X \subseteq V$ be as in the lemma and
suppose that $\sigma \colon X \to D$ has positive support in $\lambda$.
Let $y \in Y_i \cap Y_r$.
If $y = x_a$ for some $a \in D'$‚ then
$\alpha_i^{\sigma_i}(y) = \alpha_r^{\sigma_r}(y) = \gamma(x_a) = a$.
Otherwise, $y \in X_i \cap X_r$, so
$\alpha_i^{\sigma_i}(y) = \sigma_i(y) = \sigma(y) = \sigma_r(y) = \alpha_r^{\sigma_r}(y)$.

It follows that Lemma~\ref{lem:sa-reduc} is applicable, so
$\Gamma' \cup \mathcal{C}_{D'} \reducesto \Gamma' \cup \Delta$.
\end{proof}

%
%
\section{Gap Instances for SA-relaxations of $\VCSP(E_{\mathcal{G},3})$}
\label{sec:gap}

In this section, we give a construction of gap instances for SA-relaxations of
$\VCSP(E_{\mathcal{G},3})$, which shows that $E_{\mathcal{G},3}$ does
not have bounded valued relational width.
This result can also be derived from results in~\cite{Schoenebeck08:focs}
using additional non-trivial results.
We provide here a direct, elementary proof for constant level LP relaxations,
whereas~\cite{Schoenebeck08:focs} deals with linear level SDP relaxations.

Let $\mathcal{G}$ be an Abelian group over a finite set $G$ and
let $g$ be a non-zero element in $G$.
Let $R_0 = \{ (x,y,z) \in G^3 \mid x = y + z + 0\}$ and $R_g = \{ (x,y,z) \in G^3 \mid x = y + z + g \}$
and $\Delta = \{R_0, R_g\}$.
Both $R_0$ and $R_g$ are expressible in $E_{\mathcal{G},3}$:
$R_0(x,y,z)=\min_{y',z'}(R^3_0(x,y',z')+R^2_0(y',y)+R^2_0(z',z))$ and
$R_g(x,y,z)=\min_{y',z'}(R^3_g(x,y',z')+R^2_0(y',y)+R^2_0(z',z))$.
By Theorem~\ref{thm:reductions}(\ref{red:express}), it suffices to prove that $\Delta$ does not have bounded valued relational width.

Let $k \geq 3$.
We construct an unsatisfiable instance $I$ of $\VCSP(\Delta)$ and a feasible
solution to its SA$(k,k)$-relaxation.
The construction is similar to the one in~\cite[Theorem~31]{Feder98:monotone} where
it is used to show that constraint languages without ``the ability to count''
do not have bounded width.
Our theorem is a strengthening of this result. 

Let $n \geq 1$ be a positive integer.
Let $T_{n \times n}$ be the torus grid graph on $n \times n$ vertices
resulting from taking the square grid graph on $(n+1) \times (n+1)$ vertices
and identifying the topmost with the bottommost vertices as well as the 
leftmost with the rightmost vertices.

\begin{figure}[ht]
\begin{center}
\begin{tikzpicture}
[
scale = 5.2,
label distance = 3pt,
point/.style = {draw, circle, fill=black, inner sep=2pt},
edge/.style = {line width=0.5pt}
]

\foreach \y [evaluate=\y as \yn using {int(mod(\y,3))}] in {0, 1, 2, 3} {
  \foreach \x [evaluate=\x as \xn using {int(mod(\x,3))}] in {0, 1, 2, 3}
    \node (X\x\y) at (0.5*\x,4-0.5*\y) [point, label={above left:$x_{\yn,\xn}$}] {};
}

\foreach \y in {0, 1, 2, 3}
  \draw[edge] (X0\y) -- (X3\y);

\foreach \y [evaluate=\y as \yn using {int(mod(\y,3))}] in {0, 1, 2, 3}
  \foreach \x in {0, 1, 2} {
    \node (Y\x\y) at ($(X\x\y)+(0.25,0.05)$) [label={below:$y_{\yn,\x}$}] {};
}

\foreach \x in {0, 1, 2, 3}
  \draw[edge] (X\x0) -- (X\x3);

\foreach \x [evaluate=\x as \xn using {int(mod(\x,3))}] in {0, 1, 2, 3}
  \foreach \y in {0, 1, 2} {
    \node (Z\x\y) at ($(X\x\y)+(-0.05,-0.25)$) [label={right:$z_{\y,\xn}$}] {};
}

\end{tikzpicture}
\end{center}
\caption{Variables in the torus $T_{3,3}$ obtained from the $4 \times 4$ grid graph.}
\label{fig:tnn}
\end{figure}
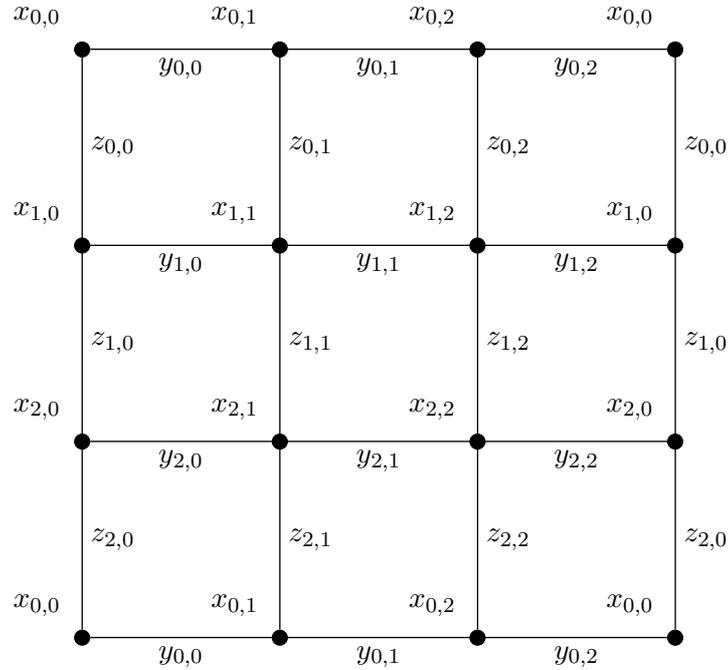

The instance $I_n$ contains one variable for each vertex and one variable for each edge in 
$T_{n \times n}$.
For $0 \leq a, b < n$,
let $x_{a,b}$, $y_{a,b}$, and $z_{a,b}$, be the variables corresponding to 
vertices,  horizontal edges, and vertical edges, respectively;
cf., Figure~\ref{fig:tnn}.
Let $I_n$ contain the following constraints:
\begin{align}
y_{a,b+1} &= y_{a,b} + x_{a,b} + c_{a,b} \label{eq:c-eqs} \\
z_{a+1,b} &= z_{a,b} + x_{a,b} + d_{a,b}, \label{eq:d-eqs}
\end{align}
where indices are taken modulo $n$, and
the elements $c_{a,b}, d_{a,b} \in \{0,g\}$ are chosen so that
\begin{equation}\label{eq:diffg}
\sum_{a,b} c_{a,b} - \sum_{a,b} d_{a,b} = g.
\end{equation}

The following result establishes Theorem~\ref{thm:eg3notbw}.
We note that it actually shows that $I_n$, which has $p=O(n^2)$ variables, is a gap instance for
SA$(k(p),k(p))$, with $k = \Theta(\sqrt{p})$.

\begin{theorem}\label{thm:in}
For every $k \geq 3$ and $n > 2k$, the instance $I_n$ is a gap instance for SA$(k,k)$.
\end{theorem}

\begin{proof}

The instance $I_n$ is unsatisfiable by construction:
Summing the equations $(\ref{eq:c-eqs})$ over $a$ and $b$
and simplifying implies the equation $0 = \sum_{a,b} (x_{a,b} + c_{a,b})$.
Similarly, the equations $(\ref{eq:d-eqs})$ imply $0 = \sum_{a,b} (x_{a,b} + d_{a,b})$.
By taking the difference of these two equations,
it follows that $0 = \sum_{a,b} (c_{a,b}-d_{a,b}) = g$ by $(\ref{eq:diffg})$,
a contradiction. Hence, the constraints of $I_n$ cannot be simultaneously satisfied.
On the other hand, the SA$(k,k)$-relaxation of $I_n$ has a feasible solution by
Lemma~\ref{lem:sakkfeas}.
\end{proof}

In the remaining part of the section, we prove that
the SA$(k,k)$-relaxation of $I_n$ has a feasible solution.

Denote by $V$ the set all variables of $I_n$ 
and let $V_x = \{ x_{a,b} \mid 0 \leq a, b < n \}$.
For $S \subseteq V_x$,
we say that $S$ \emph{excludes a cross} if
there are indices $a'$ and $b'$ such that $x_{a',b} \not\in S$ for all $0 \leq b < n$, and
$x_{a,b'} \not\in S$ for all $0 \leq a < n$.
We say that $S$ \emph{contains a hole} if
the induced subgraph $T_{n \times n}[V_x \setminus S]$ is not connected.
Let $\mathcal{S}$ be the family of subsets $S \subseteq V_x$ such that
$S$ excludes a cross and does not contain a hole.

For a subgraph $T'$ of $T_{n \times n}$, we denote by $\Var(T')$ the set of
variables on the vertices and edges of $T'$.
Let $X_1, \dots, X_m$ be an enumeration of all subsets $X \subseteq V$ such that 
$X \subseteq \Var(T_{n \times n}[S])$ for some $S \in \mathcal{S}$.
For $i\in[m]$, define
\begin{equation}\label{defbar}
\bar{X}_i = \bigcap_{S \in \mathcal{S} \colon X_i \subseteq \Var(T_{n \times n}[S])} \Var(T_{n \times n}[S]).
\end{equation}
Since $\mathcal{S}$ is closed under intersection,
it follows that $\bar{X}_i = \Var(T_{n \times n}[S])$ for some $S \in \mathcal{S}$,
so $V_x \cap \bar{X}_i = S$ excludes a cross and does not contain a hole.

It follows from the definition of $\bar{X}$  that 
$X_j \subseteq X_i \implies \bar{X}_j \subseteq \bar{X}_i$.
However, we will need a stronger property, namely that it is possible to move within the set family 
$\{ \bar{X}_i \}_{i=1}^m$ 
from $\bar{X}_j$ to $\bar{X}_i$ 
by adding vertices from $V_x$ one at a time.
More formally,
define the binary relation $\to$ on $\{ \bar{X}_i \}_{i=1}^m$ by letting $\bar{X}_j \to \bar{X}_i$
if, and only if, $\bar{X}_j \subseteq \bar{X}_i$ and $V_x \cap (\bar{X}_i \setminus \bar{X}_j) = \{x_{a,b}\}$ for some $0 \leq a,b < n$.
Let $\preccurlyeq$ be the reflexive transitive closure of $\to$.

\begin{lemma}\label{lem:reachable}
$X_j \subseteq X_i \implies \bar{X}_j \preccurlyeq \bar{X}_i$. 
\end{lemma}

\begin{proof}
Assume to the contrary that there are $i$ and $j$ such that $X_j \subseteq X_i$ but
$\bar{X}_j \not\preccurlyeq \bar{X}_i$.
Let $C = V_x \cap (\bar{X}_i \setminus \bar{X}_j)$ and assume that
$i$ and $j$ are chosen so that $\left|C\right|$ is minimised.
Let $x_{a,b} \in C$ and consider the set $S = (V_x \cap \bar{X}_i) \setminus \{x_{a,b}\}$.
By construction $\bar{X}_j \subseteq \bar{S} \subseteq \bar{X}_i$. If $x_{a,b}
\not\in \bar{S}$, then $|V_x \cap (\bar{S} \setminus \bar{X}_j)| <
\left|C\right|$ which contradicts the minimality of $\left|C\right|$. Therefore,
$V_x \cap \bar{S} = V_x \cap \bar{X}_i$, so $\bar{S} = \bar{X}_i$.
This means that $S$ contains a hole, and in particular that
$S$ and therefore $\bar{X}_i$ contains all neighbours of $x_{a,b}$.
Consider the set $\partial C$ of vertices in $V_x \setminus C$ that are neighbours to some
$x_{a,b} \in C$.
By the previous remark, $\partial C \subseteq (V_x \cap \bar{X}_i) \setminus C$,
so $\partial C \subseteq V_x \cap \bar{X}_j$.
Let $C'$ be the vertices of an excluded cross in $\bar{X}_i$.
Then, any path in $T_{n \times n}$ from a vertex in $C$ to a vertex in $C'$ must pass through
a vertex in $\partial C$.
Therefore, the induced subgraph
$T_{n \times n}[V_x \setminus \bar{X}_j]$ is disconnected, so
$V_x \cap \bar{X}_j$ contains a hole; a contradiction.
\end{proof}

For $i\in[m]$,
define $N_i$ to be the set of assignments $\bar{\sigma} \colon \bar{X}_i \to G$ that satisfy
every constraint in $I_n$ whose scope is contained in $\bar{X}_i$.
We argue that $N_i$ is non-empty for every $i$.
A \emph{horizontal component} of $\bar{X}_i$ is a set of
edges $\{ y_{a,b}, y_{a,b+1}, \dots, y_{a,b+r} \} \subseteq \bar{X}_i$
such that $y_{a,b-1}, y_{a,b+r+1} \not\in \bar{X}_i$.
A \emph{vertical component} of $\bar{X}_i$ is defined analogously.
Let
$C_i$, $H_i$, and $V_i$ be the number of 
vertices,
horizontal components, and
vertical components, respectively, in $\bar{X}_i$.
Since $V_x \cap \bar{X}_i$ excludes a cross,
an assignment is precisely determined by freely choosing the value of every vertex,
and of one edge in each horizontal component and one edge in each vertical component:
\begin{equation}\label{eq:nj}
\left| N_i \right| = \left| G \right| ^{C_i+H_i+V_i} \geq 1.
\end{equation}

For $\bar{\tau} \in N_j$ and $i$ such that $\bar{X}_j \subseteq \bar{X}_i$,
let $N_{j,i}(\bar{\tau})$ denote the set of assignments $\bar{\sigma} \in N_i$ such that
$\bar{\tau} = \Crestrict{\bar{\sigma}}{\bar{X}_j}$,
i.e.\ the set of extensions of $\bar{\tau}$ to an assignment in $N_i$.
Next, we give an expression for the size of the sets $N_{j,i}(\bar{\tau})$ that is independent
of the choice of $\bar{\tau}$.

\begin{lemma}\label{lem:CHV}
For $X_j \subseteq X_i$ and all $\bar{\tau} \in N_j$,
\begin{equation}\label{eq:CHV}
\left| N_{j,i}(\bar{\tau}) \right| = \frac{\left|N_i\right|}{\left|N_j\right|}.
\end{equation}
\end{lemma}

\begin{proof}
First assume that $\bar{X}_j \to \bar{X}_i$ and
let $x_{a,b}$ be the unique vertex in $\bar{X}_{i} \setminus \bar{X}_{j}$.
Since $\bar{X}_j \cap V_x$ does not contain a hole, it follows that $x_{a,b}$ 
must have fewer than four neighbours in $\bar{X}_j$.
We consider the following three possible cases:
\begin{enumerate}
\item
$x_{a,b}$ has a single neighbour in $\bar{X}_j$. Without loss of generality, assume that this
neighbour is $x_{a,b+1}$ so that $\bar{X}_i \setminus \bar{X}_j = \{x_{a,b}, y_{a,b}\}$.
Choose the value of $x_{a,b}$ arbitrarily.
If $y_{a,b+1} \in \bar{X}_j$, then the equation $y_{a,b+1} = y_{a,b} + x_{a,b} + c_{a,b}$
forces the value of $y_{a,b}$.
In this case, we have $|G|$ possible extensions and $C_i = C_j+1$‚ $H_i = H_j$, and $V_i = V_j$, so (\ref{eq:CHV}) holds.
Otherwise, $y_{a,b+1} \not\in \bar{X}_j$, so the value of $y_{a,b}$ can be chosen arbitrarily.
In this case, we have $|G|^2$ possible extensions and
$C_i = C_j+1$, $H_i = H_j+1$, and $V_i = V_j$, so (\ref{eq:CHV}) holds.
\item
$x_{a,b}$ has two neighbours in $\bar{X}_j$.
If $x_{a,b}$ has one horizontal and one vertical neighbour, then we can argue
as in case (1).
Otherwise, without loss of generality, assume that 
$\bar{X}_i \setminus \bar{X}_j = \{y_{a,b-1},x_{a,b},y_{a,b}\}$.
We have three possible cases,
depending on the size of the intersection $\{y_{a,b-2}, y_{a,b+1}\} \cap \bar{X}_j$.
If this intersection contains both $y$-variables,
then the values of $y_{a,b-1}$, $x_{a,b}$, and $y_{a,b}$ are all forced by
the equations.
In this case we have $1$ possible extension,
$C_i = C_j+1$, $H_i = H_j-1$, and $V_i = V_j$, so (\ref{eq:CHV}) holds.
If the intersection contains one or zero $y$-variables, then we can choose the
value of $x_{a,b}$ arbitrarily and proceed similarly to case (1).
\item
$x_{a,b}$ has three neighbours in $\bar{X}_j$.
This case follows by extending the argument in (2) for two vertical neighbours.
\end{enumerate}

We now prove by induction that the general expression in (\ref{eq:CHV}) holds.
By Lemma~\ref{lem:reachable},
there exists an $i'$ such that $\bar{X}_{j} \preccurlyeq \bar{X}_{i'} \to \bar{X}_i$.
We have just shown that (\ref{eq:CHV}) holds for $i'$, $i$, and all $\bar{\sigma}' \in N_{i'}$.
Assume by induction that (\ref{eq:CHV}) holds for $j$, $i'$, and all $\bar{\tau} \in N_j$.
Then,
\begin{align*}
\left| N_{j,i}(\bar{\tau}) \right| &= 
\sum_{\bar{\sigma}' \in N_{j,i'}(\bar{\tau})} \left| N_{i',i}(\bar{\sigma}') \right| 
= \sum_{\bar{\sigma}' \in N_{j,i'}(\bar{\tau})} \frac{\left| N_i \right|}{\left| N_{i'} \right|}
= \frac{\left| N_{i'} \right|}{\left| N_{j} \right|} \frac{\left| N_i \right|}{\left| N_{i'} \right|}
= \frac{\left| N_i \right|}{\left| N_{j} \right|}.
\end{align*}
which proves the lemma.
\end{proof}

We are now ready to finish the proof of Theorem~\ref{thm:in}.

\begin{lemma}\label{lem:sakkfeas}
For $i\in[m]$, with $\left| X_i \right| \leq k$, let $\lambda_i$ be the following probability
distribution:
\begin{equation}\label{eq:lambdadef}
\lambda_i(\sigma) = \Pr_{\bar{\sigma} \sim U_i} \left[ \Crestrict{\bar{\sigma}}{X_i} = \sigma \right],
\end{equation}
where $U_i$ is the uniform distribution on $N_i$.
Then, $\lambda$ is a feasible solution to the SA$(k,k)$-relaxation of $I_{n}$.
\end{lemma}

\begin{proof}
Let $X \subseteq V$ with $\left| X \right| \leq k$.
Since $|X| \leq k < n/2$, by the pigeonhole principle, there exists an $a'$ such that
$\{y_{a',b}, x_{a',b}, y_{a'+1,b} \} \cap X = \emptyset$ for every $0 \leq b < n$.
Similarly, there exists a $b'$ such that
$\{z_{a,b'}, x_{a,b'}, z_{a,b'+1} \} \cap X = \emptyset$ for every $0 \leq a < n$.
Let $S = V_x \setminus \{ (a,b) \mid a = a' \text{ or } b = b' \}$.
Then, $S \in \mathcal{S}$ and $X \subseteq \Var(T_{n \times n}[S])$, so $X = X_i$ for some $1 \leq i \leq m$.
It follows that $\lambda$ is defined for all $X \subseteq V$ with $\left| X \right| \leq k$.

By construction, $\lambda$ satisfies (\ref{sa:sum1}) and (\ref{sa:infeas}) for the
SA$(k,k)$-relaxation of $I_n$.
It remains to show that it also satisfies (\ref{sa:marginal}).

Let $X_j \subseteq X_i$ and $\tau \colon X_j \to G$.
Let $X$ be a subset of variables such that $X_j \subseteq X \subseteq \bar{X}_i$.
Then,
\begin{align}
\Pr_{\bar{\sigma} \sim U_i} \left[ \Crestrict{\bar{\sigma}}{X_j} = \tau \right] &=
\sum_{\sigma \colon X \to G}
\Pr_{\bar{\sigma} \sim U_i} \left[ \Crestrict{\bar{\sigma}}{X_j} = \tau \text{ and } \Crestrict{\bar{\sigma}}{X} = \sigma\right] \notag \\ 
&=
\sum_{\substack{\sigma \colon X \to G \\ \Crestrict{\sigma}{X_j} = \tau}} \Pr_{\bar{\sigma} \sim U_i} \left[ \Crestrict{\bar{\sigma}}{X} = \sigma \right]. \label{eq:conditional}
\end{align}
For $X = \bar{X}_j$, equation (\ref{eq:conditional}) implies the following.
\begin{equation}
\Pr_{\bar{\sigma} \sim U_i} \left[ \Crestrict{\bar{\sigma}}{X_j} = \tau \right] =
\sum_{\substack{\bar{\tau} \colon \bar{X}_j \to G \\ \Crestrict{\bar{\tau}}{X_j} = \tau}} \Pr_{\bar{\sigma} \sim U_i} \left[ \Crestrict{\bar{\sigma}}{\bar{X}_j} = \bar{\tau} \right] =
\sum_{\substack{\bar{\tau} \in N_j \\ \Crestrict{\bar{\tau}}{X_j} = \tau}} \frac{\left| N_{j,i}(\bar{\tau}) \right|}{\left| N_i \right|} =
\sum_{\substack{\bar{\tau} \in N_j \\ \Crestrict{\bar{\tau}}{X_j} = \tau}} \frac{1}{\left| N_j \right|} =
\lambda_j(\tau),
\end{equation}
where the next-to-last inequality follows from Lemma~\ref{lem:CHV}.
Hence,
\begin{align}
\lambda_j(\tau) 
&= \Pr_{\bar{\sigma} \sim U_i} \left[ \Crestrict{\bar{\sigma}}{X_j} = \tau \right] \label{eq:deriv2} \\
&= \sum_{\substack{\sigma \colon X_i \to G \\ \sigma|_{X_j} = \tau}} \Pr_{\bar{\sigma} \sim U_i} \left[ \Crestrict{\bar{\sigma}}{X_i} = \sigma \right] \label{eq:deriv3} \\
&= \sum_{\substack{\sigma \colon X_i \to G \\ \sigma|_{X_j} = \tau}} \lambda_i(\sigma), \label{eq:deriv4}
\end{align}
where $(\ref{eq:deriv3})$ follows by (\ref{eq:conditional}) with $X = X_i$.
It follows that $\lambda$ satisfies (\ref{sa:marginal}), hence it is a feasible solution
to the SA$(k,k)$-relaxation of $I_n$.
\end{proof}

%
%

\section{Proof of Lemma~\ref{lem:corewnus}}
\label{sec:corewnus}

\begin{lemma*}[Lemma~\ref{lem:corewnus} restated]
Let $\Gamma$ be a valued constraint language of finite size on domain $D$ and
$\Gamma'$ a core of $\Gamma$ on domain $D' \subseteq D$.
Then, $\supp(\Gamma)$ satisfies the BWC if, and only if, $\supp(\Gamma' \cup \mathcal{C}_{D'})$
satisfies the BWC.
\end{lemma*}

\begin{proof}
Let $\mu$ be a unary fractional polymorphism of $\Gamma$ with an operation $g$
in its support such that $g(D) = D'$.
We begin by constructing a unary fractional polymorphism $\mu'$ of $\Gamma$ such that
\emph{every} operation in $\supp(\mu')$ has an image in $D'$.
We will use a technique for generating fractional polymorphisms described
in~\cite[Lemma 10]{ktz15:sicomp}.
It takes a fractional polymorphism, such as $\mu$, a set of
\emph{collections} $\mathbb{G}$,
which in our case will be the set of operations in the clone of $\supp(\mu)$, a set of \emph{good}
collections $\mathbb{G^*}$, which will be operations from $\mathbb{G}$ with an image in $D'$,
and an \emph{expansion operator} {\sf Exp} which assigns to every collection 
a probability distribution on $\mathbb{G}$.

The procedure starts by generating each collection $f \in \supp(\mu)$ with probability $\mu(f)$,
and subsequently the expansion operation {\sf Exp} maps 
$f \in \mathbb{G}$ to the probability distribution 
that assigns probability $\Pr_{h \sim \mu} [h \circ f = f']$ to each operation $f' \in \mathbb{G}$.
The expansion operator is required to be \emph{non-vanishing}, which means that
starting from any collection $f \in \mathbb{G}$, repeated expansion must assign
non-zero probability to a good collection in $\mathbb{G}^*$.
In our case, this is immediate, since starting from a collection $f$, the good collection
$g \circ f$ gets probability at least $\mu(g)$ which is non-zero by assumption.
By \cite[Lemma~10]{ktz15:sicomp}, it now follows that $\Gamma$ has a fractional
polymorphism $\mu'$ with $\supp(\mu') \subseteq \mathbb{G}^*$.
So every operation in $\supp(\mu')$ has an image in $D'$.

Now, we show that if $\supp(\Gamma)$ contains an $m$-ary WNU $t$, then
$\supp(\Gamma' \cup \mathcal{C}_{D'})$ also contains an $m$-ary WNU.
Let $\omega$ be a fractional polymorphism of $\Gamma$ with $t$ in its support.
Define $\omega'$ by $\omega'(f') = \Pr_{h \sim \mu', f \sim \omega} [h \circ f = f']$.
Then, $\omega'$ is a
fractional polymorphism of $\Gamma$ in which every operation has an image in $D'$,
so $\omega'$ is a fractional polymorphism of $\Gamma'$.
Furthermore, for any unary operation $h \in \supp(\mu')$, $h \circ t$ is again a WNU, so
$\supp(\Gamma')$ contains an $m$-ary WNU $t'$.
Next, let $h(x) = t'(x, \dots, x)$.
Since $\Gamma'$ is a core, the set of unary operations in $\supp(\Gamma')$ contains only 
bijections and is closed under composition (Lemma~\ref{lem:suppclone}).
It follows that $h$ has an inverse $h^{-1} \in \supp(\Gamma')$, and since $\supp(\Gamma')$
is a clone, $h^{-1} \circ t'$ is an idempotent WNU
in $\supp(\Gamma')$.
We conclude that $h^{-1} \circ t' \in \supp(\Gamma' \cup \{ \mathcal{C}_{D'} \})$.

For the opposite direction, let $t'$ be an $m$-ary WNU in $\supp(\Gamma' \cup \{ \mathcal{C}_{D'} \})$,
and let $\omega'$ be a fractional polymorphism of $\Gamma' \cup \{ \mathcal{C}_{D'} \}$
with $t'$ in its support.
Then, $\omega'$ is also a fractional polymorphism of $\Gamma'$.
Define $\omega$ by $\omega(f) = \Pr_{h \sim \mu', f' \sim \omega'} [f'[h, \dots, h] = f]$.
Then, $\omega$ is a fractional polymorphism of $\Gamma$, and, for every $h \in \supp(\mu')$,
the operation $t[h, \dots, h]$ is an $m$-ary WNU in $\supp(\omega)$.
We conclude that $t \in \supp(\Gamma)$,
which finishes the proof. 
\end{proof}

%
%

\section{Proofs of Theorems~\ref{thm:injective} and~\ref{thm:cons2}}

\begin{theorem*}[Theorem~\ref{thm:injective} restated]
Let $D$ be an arbitrary finite domain and let $\Gamma$ be an arbitrary valued constraint
language of finite size on $D$ with $\mathcal{C}_D\subseteq \Gamma$. Assume that
$\Gamma$ expresses a unary finite-valued weighted relation $\nu$ that is injective on $D$.
Then, either $\supp(\Gamma)$ satisfies the BWC, in which case $\Gamma$ has
valued relational width $(2,3)$, or $\VCSP(\Gamma)$ is NP-hard. 
\end{theorem*}

\begin{proof}
If $\Gamma$ satisfies the BWC then the result follows from
Theorem~\ref{thm:mainicalp}. If $\Gamma$ does not satisfy the BWC then,
by Lemma~\ref{lem:crispwnus}, there exists a crisp constraint language $\Delta$ 
such that $\pol(\Delta)$ does not satisfy the BWC and $\Delta \reducesto \Gamma$.
By assumption, $\mathcal{C}_D\subseteq \Gamma$ and thus $\Delta\cup\mathcal{C}_D
\reducesto \Gamma$. Hence we may assume, without loss of generality, that
$\mathcal{C}_D\subseteq\Delta$.
By Theorem~\ref{thm:clarity}, there exists a non-trivial Abelian group
$\mathcal{G}$ over a finite set $G$ and an interpretation of $E_{\mathcal{G},3}$
in $\Delta$ with parameters $(d,S,h)$. 
By Theorem~\ref{thm:reductions}(\ref{red:interpret}),
we have $E_{\mathcal{G},3} \reducesto \Delta \reducesto \Gamma$.

Let $C$ be larger than $\max_{a \in D} \nu(a) - \min_{a \in D} \nu(a)$. The $d$-ary
weighted relation $\phi(x_1,\ldots,x_d)$ defined by
\[
\phi(x_1,\ldots,x_d)=\nu(x_1)+C\nu(x_2)+C^2\nu(x_3)+\dots+C^{d-1}\nu(x_d)
\]
is injective on the set of $d$-tuples over $D$. 
By Theorem~\ref{thm:reductions}(\ref{red:express}),
$\{\phi\}\reducesto\Gamma$. 
We define 
\[
\phi'(x_1,\ldots,x_d) = \min_{y_1,\ldots,y_d}
h^{-1}(\eq{G})(x_1,\ldots,x_d,y_1,\ldots,y_d) + \phi(y_1,\ldots,y_d).
\]
By Theorem~\ref{thm:reductions}(\ref{red:express}) and~(\ref{red:interpret}), 
$\{\phi'\}\reducesto\Gamma$.
Thus, we have an injective unary weighted
relation $\phi'$ on the interpreted $E_{\mathcal{G},3}$.
For every $x\in G$, let $h_x\in D^d$ be an arbitrarily chosen element of $h^{-1}(x)$.
Finally, define the unary finite-valued weighted relation $\phi'':G\to\mathbb{Q}$ by
$\phi''(x)=\phi'(h_x)$. (Note that the choice of $h_x$ does not affect the value
of $\phi''(x)$.)

We denote by $E'_{\mathcal{G}}$ the crisp constraint language on domain $G$ with, for every $r\geq 1$, $a\in G$, and
$\vec{c}=(c_1,\ldots,c_r)\in\mathbb{Z}^r$ with $\sum_{i=1}^r c_i=0$, a relation
$S^r_{a,\vec{c}}=\{(x_1,\ldots,x_r)\in G^r \mid \sum_{i=1}^r c_ix_i=a\}$.
By~\cite[Theorem~3.18]{Thapper10:thesis}, $\VCSP(E'_{\mathcal{G}}\cup\{\phi''\})$
is APX-hard, and thus NP-hard since $\phi''$ is injective and thus non-constant
on $G$. We will finish the proof by showing how to reduce, in polynomial time,
any instance $I'$ of $\VCSP(E'_{\mathcal{G}}\cup\{\phi''\})$ to an instance $I$ of
$\VCSP(E_{\mathcal{G},3}\cup\{\phi''\})$.

Let $V$ denote the set of variables of $I'$.
The variables of $I$ will include $V$ and a set of new auxiliary variables for
each constraint of $I'$ not involving $\phi''$.
Let $\phi''(x)$ be a constraint of $I'$ for some $x\in V$. Then we include
the constraint $\phi''(x)$ in $I$.
Let $S^r_{a,\vec{c}}(\vec{x})$ be a constraint of
$I'$ for some $r\geq 1$, $a\in G$, $\vec{c}=(c_1,\ldots,c_r)\in\mathbb{Z}^r$
with $\sum_{i=1}^r c_i=0$, and $\vec{x}=(x_1,\ldots,x_r) \in V^r$ . Since
$|G|x=0$ in $\mathcal{G}$, for all $x\in G$ we can, without loss of generality, assume that
$0\leq c_i<|G|$. Thus $S^r_{a,\vec{c}}$ is equivalent to an $m$-ary relation $S'$ over $G$
where $m=\sum_{i=1}^r c_i\leq r|G|$.
The relation $S'$ can be expressed with $O(m)$ relations from $E_{\mathcal{G},3}$ using
$O(m)$ auxiliary variables.
\end{proof}

\begin{theorem*}[Theorem~\ref{thm:cons2} restated]
Let $\Gamma$ be a conservative valued constraint language.
Either $\VCSP(\Gamma)$ is NP-hard, or $\supp(\Gamma)$ contains a majority operation
and hence $\Gamma$ has valued relational width $(2,3)$.
\end{theorem*}

\begin{proof}
If $\pol(\Gamma)$ does not contain a majority operation then $\Gamma$ is NP-hard
by Theorem~\ref{thm:takhanov-kz}. If $\supp(\Gamma)$ contains a majority operation
then, by Corollary~\ref{cor:maj}, $\Gamma$ has valued relational width $(2,3)$.

Let $F$ be the set of majority operations in $\pol(\Gamma) \setminus
\supp(\Gamma)$. By Lemma~\ref{lem:killing}, for each $f \in F$, there is an
instance $I_f$ of $\VCSP(\Gamma)$ such that $f \not\in \pol(\opt(I_f))$. Let
$\Gamma' = \Gamma \cup \{ \opt(I_f) \mid f \in F \}$. If $\pol(\Gamma')$ does
not contain a majority polymorphism, then, since $\Gamma$ is conservative, so is
$\Gamma'$, and hence $\Gamma'$ is NP-hard by Theorem~\ref{thm:takhanov-kz}.
Therefore, $\Gamma$ is NP-hard by
Theorem~\ref{thm:reductions}~(\ref{red:feasopt}). Assume that $\pol(\Gamma')$
contains a majority polymorphism $f$. Then, $f \not\in F$, so $f \in
\supp(\Gamma)$. From Corollary~\ref{cor:maj}, it follows that $\Gamma$ has
valued relational width $(2,3)$. 
\end{proof}

%
%
\section{Conclusions}
Using techniques from the algebraic study of CSPs and the study of linear
programming relaxations, we have given a precise characterisation of the power
of constant level Sherali-Adams linear programming relaxations for exact
solvability of valued constraint languages. Notably, we needed to prove that
certain gadget constructions, such as going to the core and interpretations, common in
the algebraic CSP literature but not commonly used in other areas of CSPs, such
as approximation, preserve solvability by constant level Sherali-Adams
relaxations.

The complexity of Min-Ones problems with respect to exact solvability and
approximability was established
in~\cite{Creignouetal:siam01,Khanna00:approximability}. Minimum-Solutions
problems are a generalisation of Min-Ones problems to larger domains, including
integer programs over bounded domains~\cite{Jonsson08:siam}. Following our
characterisation of the power of Sherali-Adams, we have given a complete
complexity classification of exact solvability of Minimum-Solution problems over
arbitrary finite domains.

\section*{Acknowledgements}
The authors wish to thank Manuel Bodirsky for explaining the relation between
varieties, pseudovarieties, and pp-interpretations in the context of finite
domain CSPs and the anonymous referees for their careful reading of an earlier
version of this paper. The authors are also grateful to Albert Atserias for useful
discussions and for bringing~\cite{Schoenebeck08:focs} to our attention.


\newcommand{\noopsort}[1]{}\newcommand{\Zivny}{\noopsort{ZZ}\v{Z}ivn\'y}

\end{document}